\newif\iflipics
\lipicstrue

\iflipics

\documentclass[a4paper,UKenglish,cleveref,autoref,cequ,noclassification]{lipics-v2021}

\pdfoutput=1 
\hideLIPIcs  

\graphicspath{{./graphics/}}

\title{Flipping and Forking}

\author{Wojciech Przybyszewski}{University of Warsaw, Warsaw, Poland \and \url{https://www.mimuw.edu.pl/~przybyszewski/}}{przybyszewski@mimuw.edu.pl}{https://orcid.org/0000-0003-1158-9925}{}
\author{Szymon Toruńczyk}{University of Warsaw, Warsaw, Poland \and \url{https://www.mimuw.edu.pl/~szymtor/}}{szymtor@mimuw.edu.pl}{https://orcid.org/0000-0002-1130-9033}{}

\authorrunning{W. Przybyszewski and S. Toruńczyk} 

\Copyright{Wojciech Przybyszewski and Szymon Toruńczyk} 

\ccsdesc[100]{} 
\keywords{structural graph theory,
model theory,
finite model theory,
monadic stability,
forking independence} 

\funding{\flag[40px]{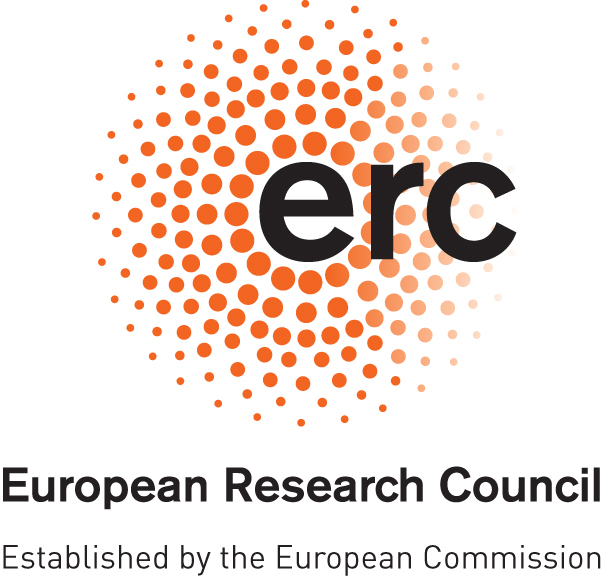}\flag[40px]{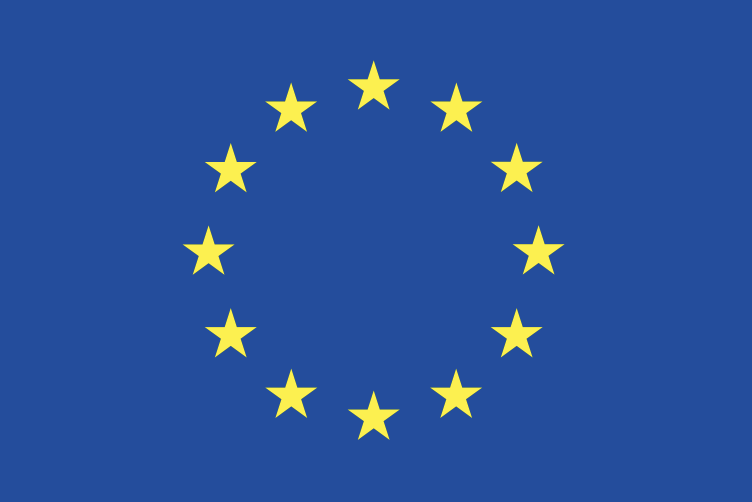}WP was supported by project {\sc{bobr}} that has received funding from the European Research Council (ERC) under the European Union's Horizon 2020 research and innovation programme (grant No 948057). ST was supported by ERC project {\sc{buka}} (grant No 101126229).}

\acknowledgements{We are grateful to Joanna Fijalkow for many valuable discussions on flips of relational structures.}

\nolinenumbers 

\EventEditors{John Q. Open and Joan R. Access}
\EventNoEds{2}
\EventLongTitle{42nd Conference on Very Important Topics (CVIT 2016)}
\EventShortTitle{CVIT 2016}
\EventAcronym{CVIT}
\EventYear{2016}
\EventDate{December 24--27, 2016}
\EventLocation{Little Whinging, United Kingdom}
\EventLogo{}
\SeriesVolume{42}
\ArticleNo{23}

\usepackage{comment}
\usepackage[disable]{todonotes}
\usepackage{mathtools}
\usepackage{thm-restate}
\usepackage{enumitem}
\usepackage{framed}
\usepackage{mdframed}

\usepackage{fixltx2e, ifthen, sidecap, wrapfig, xspace}
\usepackage{url}
\usepackage{hyperref}

\newcommand\sz[1]{\todo[size=\scriptsize,backgroundcolor=lime]
{\textbf{sz:}#1}}

\renewcommand{\cal}[1]{\mathcal {#1}}

\newcommand{\Cc}{{\mathcal C}}

\newcommand{\CC}{\Cc}

\newcommand{\gaif}{\mathit{Gaif}}
\newcommand{\Gg}{{\mathcal G}}
\newcommand{\GG}{\Gg}

\newcommand{\Ll}{{\mathcal L}}

\newcommand{\N}{\mathbb{N}}

\renewcommand{\preceq}{\preccurlyeq}
\renewcommand{\subset}{\subseteq}
\renewcommand{\le}{\leqslant}
\renewcommand{\ge}{\geqslant}

\renewcommand{\phi}{\varphi}
\renewcommand{\epsilon}{\varepsilon}

\renewcommand{\hat}{\widehat}

\def\xInd#1#2{#1\setbox0=\hbox{$#1x$}\kern\wd0\hbox to 0pt{\hss$#1\mid$\hss}
	\lower.9\ht0\hbox to 0pt{\hss$#1\smile$\hss}\kern\wd0}

\def\xind{\mathop{\mathpalette\xInd{}}} 

\def\xnotind#1#2{#1\setbox0=\hbox{$#1x$}\kern\wd0
	\hbox to 0pt{\mathchardef\nn=12854\hss$#1\nn$\kern1.4\wd0\hss}
	\hbox to 0pt{\hss$#1\mid$\hss}\lower.9\ht0 \hbox to 0pt{\hss$#1\smile$\hss}\kern\wd0}

\def\xnind{\mathop{\mathpalette\xnotind{}}} 
\newcommand{\ind}[2]{\xind_{#1}^{#2}}
\newcommand{\nind}[2]{\xnind_{#1}^{#2}}
\newcommand{\find}[1]{\xind_{#1}}
\newcommand{\nfind}[1]{\xnind_{#1}}
\newcommand{\Types}[1]{\mathrm{Types}^{#1}}

\newcommand{\tup}{\bar}
\newcommand{\tp}{\mathrm{tp}}
\newcommand{\atp}{\mathrm{atp}}
\newcommand{\otp}{\mathrm{otp}}
\newcommand{\fdist}{\textup{flip-dist}}

\DeclareMathOperator{\Th}{\text{Th}}

\newcommand\set[1]{\ensuremath{\{#1\}}}

\newcommand{\from}{\colon}
\newcommand{\setof}[2]{\set{#1\mid#2}}

\newcommand{\str}[1]{\mathbf{#1}}

\newcommand{\mathsym}[1]{{}}

\DeclareMathOperator{\dist}{dist}

\DeclareMathOperator{\xor}{xor}

\definecolor{gray1}{rgb}{0.99,0.99,0.99}
\definecolor{gray2}{rgb}{0.97,0.97,0.97}
\definecolor{gray3}{rgb}{0.95,0.95,0.95}
\definecolor{gray4}{rgb}{0.93,0.93,0.93}
\definecolor{gray5}{rgb}{0.91,0.91,0.91}
\definecolor{gray6}{rgb}{0.89,0.89,0.89}
\definecolor{gray7}{rgb}{0.87,0.87,0.87}
\definecolor{gray8}{rgb}{0.85,0.85,0.85}
\definecolor{gray9}{rgb}{0.83,0.83,0.83}
\definecolor{gray10}{rgb}{0.81,0.81,0.81}
\definecolor{gray20}{rgb}{0.71,0.71,0.71}
\definecolor{gray40}{rgb}{0.51,0.51,0.51}

\newlength{\leftbarwidth}
\newlength{\leftbarsep}
\renewenvironment{leftbar}[1][blue]
{%
\def\FrameCommand
{%
{\hspace{-8pt} \color{#1} \vrule width 3pt}%
\hspace{0pt}
\fboxsep=\FrameSep\colorbox{#1!5!}%
}%
\MakeFramed{\hsize\hsize\advance\hsize-\width\FrameRestore}%
}
{\endMakeFramed}
\makeatletter

  \if@todonotes@disabled 
  \excludecomment{szbar}
  \else 
  \newenvironment{szbar}{\begin{leftbar}[red]{\noindent\bf sz: }}{\end{leftbar}}
  \fi

\makeatother

\makeatletter

  \if@todonotes@disabled 
  \excludecomment{app}
  \else 
  
  \fi

\makeatother

\else
\documentclass[conference]{IEEEtran}

\usepackage[noadjust]{cite}
\usepackage{amsmath,amssymb,amsfonts}
\usepackage{algorithmic}
\usepackage{graphicx}
\usepackage{textcomp}
\usepackage{xcolor}
\usepackage[T1]{fontenc}
\usepackage{amsmath}
\usepackage{amsthm}
\usepackage{mathtools}
\usepackage{stackengine}
\usepackage{bbm}
\usepackage{thm-restate}
\usepackage{hyperref}
\usepackage{breqn}
\usepackage{cleveref}
\usepackage{wrapfig} 
\usepackage{tikz}
\usepackage{comment}
\usepackage[disable]{todonotes}
\usepackage{mathtools}
\usepackage{thm-restate}
\usepackage{enumitem}
\usepackage{framed}
\usepackage{mdframed}
\usepackage{url}
\usepackage{hyperref}

\usetikzlibrary{arrows.meta,positioning}

\input{envs_lics}

\def\BibTeX{{\rm B\kern-.05em{\sc i\kern-.025em b}\kern-.08em
    T\kern-.1667em\lower.7ex\hbox{E}\kern-.125emX}}
    
\title{Flipping and Forking}
\author{\IEEEauthorblockN{Wojciech Przybyszewski \thanks{
    \vspace*{0.2cm}
    \hspace*{-2cm}\begin{minipage}{\columnwidth}
        WP was supported by project {\sc{bobr}} that has received funding from the European Research Council (ERC) under the European Union's Horizon 2020 research and innovation programme (grant No 948057). ST was supported by ERC project {\sc{buka}} 
        (grant No 101126229).
    \end{minipage}\hfill
    \begin{minipage}{0.5\columnwidth}
        \includegraphics[width=\textwidth]{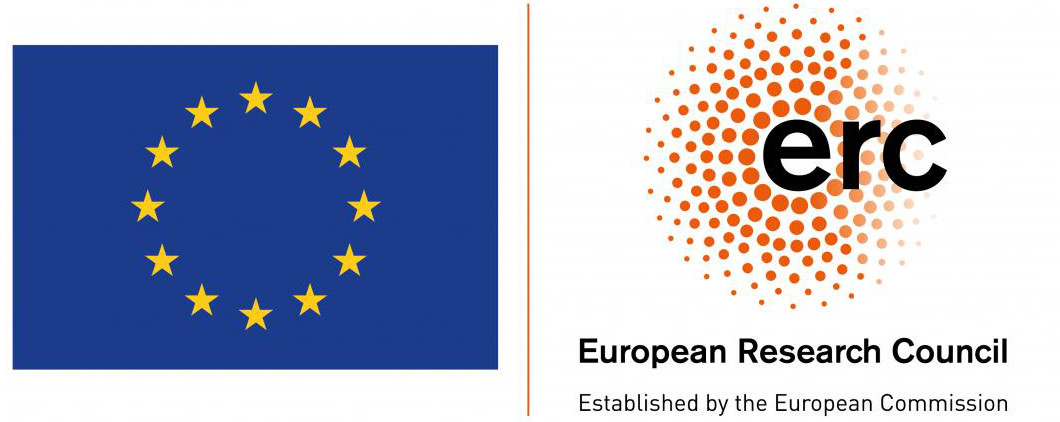}
    \end{minipage}\hfill
}}
\IEEEauthorblockA{Institue of Informatics\\
Univeristy of Warsaw\\
Warsaw, Poland\\
przybyszewski@mimuw.edu.pl}
\and
\IEEEauthorblockN{Szymon Toruńczyk}
\IEEEauthorblockA{Institue of Informatics\\
Univeristy of Warsaw\\
Warsaw, Poland\\
szymtor@mimuw.edu.pl}}
\IEEEoverridecommandlockouts

\fi

\begin{document}
\maketitle

\begin{abstract}
  Monadic stability and the more general monadic dependence (or NIP) are tameness conditions for classes of logical structures, studied in the 80's in Shelah's classification program in model theory. They recently emerged in algorithmic and structural graph theory and finite model theory as central notions in relation with the model checking problem for first-order logic: the problem was shown to be fixed-parameter tractable for inputs which come from a fixed class of graphs which is monadically stable, and is conjectured to be tractable in all monadically dependent classes. Several combinatorial characterizations of such graph classes turned out to be essential in their algorithmic treatment; they are all based on the fundamental operation of ``flipping'' a graph.

  We introduce the notions of \emph{flips} and \emph{flip independence} in arbitrary relational structures. 
  We lift prior combinatorial characterizations of monadically stable graph classes
  to monadically stable classes of relational structures.
  We show the equivalence of flip independence with \emph{forking independence} (over models) -- a logical notion of paramount importance in stability theory -- in monadically stable structures, shedding new light on the relevance of flips, also characterizing forking independence (over models) combinatorially.
  We give more precise descriptions of forking independence in the case of 
  monadically stable graphs, and relational structures with a nowhere dense Gaifman graph.
\end{abstract}

\iflipics\else
\begin{IEEEkeywords}
structural graph theory,
model theory,
finite model theory,
monadic stability,
forking independence
\end{IEEEkeywords}
\fi

\section{Introduction}

Over the past two decades, the following problem has attracted significant attention in the areas of finite and algorithmic model theory,
and algorithmic graph theory \cite[Section~8]{logic-graphs-algorithms}:
\begin{quote}\itshape
    Characterize those classes of structures for which the model checking problem for first-order logic is fixed-parameter tractable.
\end{quote}
For a fixed class $\CC$ of structures, the \emph{model-checking problem for first-order logic}  is the problem of deciding whether a given structure $\str A\in\CC$ satisfies a given  first-order sentence $\phi$.
The problem is \emph{fixed-parameter tractable} if it can be solved in  time polynomial in the size of $\str A$, where the degree of the polynomial does not depend on $\phi$.
If such an algorithm for the class $\CC$ exists, we call $\CC$ \emph{tractable} for brevity.

Most of the focus in the area has been on \emph{graph} classes, allowing to leverage  tools emerging from structural graph theory.
Tractable graph classes include the class of planar graphs, graph classes of bounded treewidth or clique-width, proper minor-closed graph classes, the class of unit interval graphs, proper hereditary classes of permutation graphs, every nowhere dense graph class \cite{gks}, 
and others. 

From the perspective of structural graph theory,
it is natural to consider \emph{hereditary} classes of graphs, or other structures -- classes, which are closed under taking induced substructures.
It is conjectured \cite{warwick-problems,tww1,flip-breakability} that for hereditary classes, tractability coincides 
with \emph{monadic dependence}.
\begin{conjecture}\label{conj:mc}
For a hereditary class $\CC$ of structures in a finite relational language,
    $$\textit{\textit{$\CC$ is tractable}}\quad\iff\quad\textit{$\CC$ is monadically dependent}.$$        
    \end{conjecture}    
Monadic dependence is a notion originating in model theory,
as part of Shelah's vast classification program \cite{Shelah1986}. 
Briefly, 
a class $\CC$ is \emph{not} monadically dependent 
if there is a formula $\phi(x,y)$ in the language 
of $\CC$ expanded by unary predicates, such that for every finite graph $G$
there is a unary expansion $\widehat{\str A}$ of some structure $\str A\in \CC$ and a subset $X\subset V(\str A)$,
such that $G$ is isomorphic to
the graph with vertices $X$ and edges 
 $\set{u,v}$ such that $\widehat{\str A}\models\phi(u,v)$
 (another, equivalent definition is given in \Cref{sec:prelims}).

For hereditary graph classes
\cite{flip-breakability},
the left-to-right implication in \Cref{conj:mc} has been recently established (under complexity-theoretic assumptions).
For general hereditary classes of graphs, the right-to-left implication in \Cref{conj:mc} remains open, but has been confirmed in several important cases, including monotone graph classes 
and more generally, edge-stable graph classes (see below),
and -- in the twin-width framework --
for classes of permutation graphs \cite{tww1}, of ordered graphs \cite{tww4}, or of tournaments \cite{tww-tournaments}.

In particular, the celebrated result of Grohe, Kreutzer, and Siebertz \cite{gks}, states
that every \emph{nowhere dense}  graph class is tractable.
This is a  structural property introduced by Ne\v set\v ril and Ossona de~Mendez \cite{NesetrilM11a}.
Nowhere dense graph classes include many well-studied classes of sparse graphs, such as the class of planar graphs, any class excluding a fixed graph as a minor, and any graph class with bounded maximum degree.


\begin{wrapfigure}{r}{0.23\linewidth} 
    \centering
    \includegraphics[width=\linewidth]{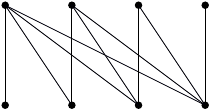}
        \caption{A half-graph of order $4$.}\label{fig:half-graph}
\end{wrapfigure}
Extending the result of Grohe, Kreutzer, and Siebertz, the following
result characterizes all  hereditary, tractable graph classes which are \emph{edge-stable}:
 avoid a \emph{half-graph} as a semi-induced\footnote{That is, for some $n\in\N$, no graph $G$ in the class contains vertices $a_1,b_1,\ldots,a_n,b_n$ with $(a_i, b_j)\in E(G)\iff i\le j$.} bipartite subgraph
(see \Cref{fig:half-graph}).

        

\begin{theorem}[\cite{rankwidth-meets-stability,ms-mc1,ms-mc2}]\label{thm:ms-mc}
    The following conditions are equivalent for a hereditary, edge-stable graph class $\CC$:
    \begin{enumerate}
        \item $\CC$ is tractable\footnote{The first item implies the other ones under a complexity theoretic assumption, $\mathrm{AW}[*]\neq \mathrm{FPT}$},
        \item $\CC$ is monadically stable,
        \item $\CC$ is monadically dependent.
    \end{enumerate}
\end{theorem}

\emph{Monadic stability} is another notion originating in model theory,
introduced by Baldwin and Shelah \cite{baldwin-shelah}.
A class $\CC$ is \emph{not} monadically stable 
if there is a formula $\phi(x,y)$ in the language 
of $\CC$ expanded by unary predicates, such that for every $n\in\N$
there is a unary expansion $\widehat{\str A}$ of a structure $\str A\in \CC$ and elements $v_1,\ldots,v_n$ of $\str A$ such that $\widehat{\str A}\models\phi(v_i,v_j)$ if and only if 
$i\le j$, for $i,j\in[n]$.
 Thus, only constant-size sets of vertices can be totally ordered by a fixed first-order formula in unary expansions of structures from a  monadically stable class $\CC$.
 Monadically stable graph classes include
 all nowhere dense classes, and their (1-dimensional) \emph{interpretations},
 e.g. the class of map graphs.



\medskip
Classes of structures equipped with relations of arity three or more haven't been explored much in the context of the model checking problem. In principle, \Cref{conj:mc} might also apply in this setting. As a first step, the following result extends a previous result known for monotone graph classes,  to classes of relational structures.
A class of relational structures is \emph{monotone} if it is closed under taking weak substructures: removing tuples from relations, and elements from the domain arbitrarily.
The \emph{Gaifman graph} of a relational structure $\str A$
is the graph whose vertices are the elements of $\str A$ and edges 
connect pairs of vertices which occur in some tuple 
in some relation of the structure.

\begin{theorem}[\cite{AdlerA14,sparsity-book,gks,nowhere-dense-structures-icalp}]\label{thm:nowhere-dense-rel}
    The following conditions are equivalent for a monotone class $\CC$ of relational structures:
    \begin{enumerate}
        \item $\CC$ is tractable\footnotemark[2],
        \item the class of Gaifman graphs of structures in $\CC$ is nowhere dense,
        \item $\CC$ is monadically stable,
        \item $\CC$ is monadically dependent.
    \end{enumerate}
\end{theorem}

For an example monadically stable class of relational structures,
 fix a first-order formula  $\phi(x_1,\ldots,x_k)$ in the language of graphs, e.g. the formula $\phi(x_1,x_2,x_3)$ expressing that $x_1,x_2,x_3$ have a common neighbor. For every graph $G$, consider the 
 structure $\phi(G)$ with domain $V(G)$ and a relation of arity $k$ 
 interpreted as $\setof{\bar a\in V(G)^k}{G\models\phi(\bar a)}$.
 If $\CC$ is a nowhere dense (or just monadically stable) graph class,
 then the class $\setof{\phi(G)}{G\in\CC}$ is also monadically stable (and usually beyond the scope of \Cref{thm:nowhere-dense-rel}).

An analogue of the equivalence (1)$\leftrightarrow$(2) in \Cref{thm:ms-mc} for classes of relational structures is not known
(here, \emph{edge-stability} is generalized by the notion 
of \emph{atomic-stability},
requiring that every atomic formula is stable; see \Cref{sec:mt-prelims}).
Note that in general, for a hereditary class $\CC$ of relational structures, the class $\gaif(\CC)$ 
does not determine whether or not $\CC$ is tractable
(for instance, if $\CC$ is an infinite class of ordered graphs -- graphs equipped with a total order on the vertex set -- then $\gaif(\CC)$ 
is always the class of all cliques).


\section{Contribution}
In this paper, we make a step towards generalizing \Cref{thm:ms-mc} to classes of relational structures.
To this end, we extend two known combinatorial characterizations 
of monadically stable graph classes, which have been pivotal in the establishment of \Cref{thm:ms-mc}, to monadically stable classes of relational structures. The prior characterizations rely on the graph-theoretic \emph{flip} operation
-- inverting the adjacency between two sets of vertices of a graph --
which has been fundamental in the  analysis of dense graphs \cite{flipper-game,flip-flatness,flip-breakability}. 
Our main contributions are as follows:
\begin{enumerate}
    \item We propose a notion of \emph{flips} for relational structures (see \Cref{def:definable-flip-intro}) and of \emph{flip-independence},
    \item We obtain the first combinatorial characterizations of monadically stable classes of relational structures, extending previous results concerning graph classes
    (see \Cref{thm:flip-flatness-rel}),
    \item We provide the first combinatorial characterization of \emph{forking independence} (over models)  -- a quintessential notion in stability theory -- in monadically stable relational structures, in terms of flips (see \Cref{thm:forking-for-structures}),
    \item We provide more refined characterizations of forking independence in monadically stable graph classes
    (see \Cref{thm:discrepancy}), and also in classes of relational structures whose Gaifman graphs form a nowhere dense graph class (see \Cref{thm:forking-nowhere-dense-structures}). 
     In addition, we characterize the latter classes, extending the result of~\cite{nowhere-dense-structures}
     (see \Cref{thm:nowhere-dense-struct}).
\end{enumerate}

We now describe those contributions in greater detail. We start with discussing the role of flips in previous results in the area, concerning graphs.

\subsection{What's in a flip?}Flips are a fundamental tool in the recent combinatorial and algorithmic study of monadically stable and monadically dependent graph classes \cite{flip-flatness, flipper-game, flip-breakability, flip-width, DBLP:conf/lics/BonnetDGKMST22}.
For a graph~$G$, \emph{flipping} a pair $X,Y$ of sets of vertices of~$G$,
results in the graph $G'$ obtained from $G$ by inverting the adjacency of pairs of different vertices in $X\times Y$.
If a graph $G'$ can be obtained from a graph $G$ by applying this operation $k$ times, we say that $G'$ is a \emph{$k$-flip} of $G$. 
Flips emerge naturally in the study of dense graphs, such as graphs of bounded clique-width, as a means for replacing the graph $G$ under consideration by another graph $G'$, which 
is somewhat ``sparser'' than $G$, yet is logically equivalent to it. 

To illustrate the use of this notion, the following result has a Ramsey-type flavor, and generalizes a property 
called \emph{uniform quasi-wideness} \cite{Dawar10} of nowhere dense graph classes.

\begin{theorem}[\cite{flip-flatness}]\label{thm:flip-flatness}
    Let $\CC$ be a monadically stable graph class. For all $r\in\N$ there is some $k\in\N$ and an unbounded function $U\from\N\to\N$ with the following property.
    For every graph $G\in\CC$ and set $A\subseteq V(G)$, there is a $k$-flip $G'$ of $G$
    and a set $B\subseteq A$ with $|B|\ge U(|A|)$,
    whose elements are pairwise at distance~${>}r$ in $G'$.
\end{theorem}

The following result can be viewed as a recursive process which locally simplifies a given graph; it is usually formulated in terms of a game between two players (c.f. \cite[Section 3.1]{flipper-game-arxiv}).
For a graph $G$ and number $r\in\N$, let 
$\textup{flip-rk}_r(G)=1$ if $G$ has one vertex,
and otherwise,
$$\textup{flip-rk}_r(G)\coloneqq 
1+\min_{\stackrel{G':}{\text{a $1$-flip of $G$}}}\max_{v\in V(G')} {\textup{flip-rk}_r(G'[B_r^{G'}(v)])},
$$
where $G'[B_r^{G'}(v)]$ denotes the subgraph of $G'$ induced by the vertices at distance at most $r$ from $v$ in $G'$.

\begin{theorem}[\cite{flipper-game}]\label{thm:fg-ranks}
    Let $\CC$ be a monadically stable graph class. For all $r\in\N$ there is some $k\in\N$ such that 
    $\textup{flip-rk}_r(G)<k$ for all $G\in\CC$.
\end{theorem}
\Cref{thm:fg-ranks} generalizes an earlier result \cite{gks} 
characterizing nowhere dense classes (in terms of \emph{Splitter games}), 
and
is crucially used in the model-checking algorithm for monadically stable classes in \Cref{thm:ms-mc} \cite{ms-mc1,ms-mc2}.
Both in \Cref{thm:flip-flatness} and in \Cref{thm:fg-ranks}, the conditions are in fact characterizations of monadic stability.


\medskip
Flips serve a similar role 
as vertex removals in the setting of sparse graphs.
For instance, a crucial notion in the analysis of planar graphs, or graphs of bounded treewidth, is that of a \emph{balanced separator}. Let us tentatively say that a set $S$ of vertices of an $n$-vertex graph $G$ is a \emph{balanced separator} 
if every connected component of $G-S$ 
has at most $\frac 23 n$ vertices. For instance, every tree has a balanced separator of size $1$. 

Balanced separators allow to recursively decompose a graph into smaller parts which interact with each other in a limited way, via~$S$.
However, balanced separators are not well suited for the study of dense graphs.
For instance, complementing the edge set of a tree $T$
results in a graph $\overline T$ with no  balanced separators of size $o(n)$.
On the other hand, $T$ and $\overline T$ have similar logical properties. Observe that $\overline T$ is a $1$-flip of $T$ (flip the pair $V(T),V(T)$).
Moreover, the graph $T^s$ obtained from $T$ by isolating (removing all edges incident to) some vertex $s$
is a $1$-flip of $T$ (flip the pair $\set{s}$ and $N(s)$ in $T$).
Therefore, if $\set{s}$ forms a balanced separator in $T$,
then $T^s$ is a $2$-flip of $\overline T$ in which every connected component 
has at most $\frac 23 n$ vertices. Because the operation of flipping 
is invertible, properties of $T^s$ can be translated 
into properties of $\overline T$, and vice-versa. This way, some properties of $\overline T$ can be reduced to the analysis of the connected components of its $2$-flip $T^s$.

In a similar way, \Cref{thm:fg-ranks}
allows to reduce \cite{ms-mc1} the problem of evaluating a given sentence in a graph $G\in\CC$, to the problem of evaluating some other sentences 
in the graphs $G'[B_r^{G'}(v)]$, for which the parameter $\textup{flip-rk}_r$ decreases by one.

\paragraph{Definable flips}
It appears that in the analysis of monadically stable (and dependent) graph classes,
it is enough to consider a very restricted form of flips, 
called \emph{definable} flips \cite{flipper-game}.
Let $G$ be a graph and $S$ be a set of its vertices.
A $k$-flip $G'$ of $G$ is \emph{$S$-definable}
if for all the involved pairs $X,Y$ of flipped sets, both $X$ and $Y$ are boolean combinations of sets of the form $\set{s}$ and $N(s)$ (the neighborhood of $s$ in $G$), for $s\in S$.
For instance, in the tree example above, $T^s$ is an $\set{s}$-definable $2$-flip of $\overline T$.
\Cref{thm:flip-flatness,thm:fg-ranks} can be strengthened by considering only $S$-definable $k$-flips, for $|S|\le k$.

\paragraph{Flip independence}
The following notion of \emph{flip-independence} has been proposed in 
\cite{flipper-game}.
Let $G$ be a graph and $S$ be a set of its vertices, and let $a,b$ be two vertices of~$G$. 
For $r\in\N$, write $$a\ind S r b$$ 
if there is some $S$-definable flip $G'$ of $G$ 
such that there is no path of length at most $r$ in $G'$ connecting $a$ and $b$, or $a=b\in S$.
Intuitively, the set $S$ controls the distance-$r$ interactions between $a$ and $b$, and can be thought of as a separator between $a$ and $b$ 
up to distance $r$.
It has been suggested \cite{flipper-game} that 
flip independence may be related to the notion of forking independence in model theory. Our first results confirm this expectations, as we now explain.

\subsection{Flipping and forking in graphs}
Forking independence is a tool of paramount importance in model theory \cite{Pillay1986-PILAIT}. 
It generalizes linear independence in vector spaces and algebraic independence in field theory.
Its definition 
is rather involved (see \Cref{def:forking} for a special case),
but the essential idea is that  for an  infinite structure $\str M$ and set $S$ of its elements, two elements $a,b$ of $\str M$ are \emph{forking independent} over $S$, written
$$a\find S b,$$ 
if the (first-order definable) relations between $a$ and $b$ 
are the same as between $a$ and a ``generic'' element
$b'$ with the same relationship to $S$ as $b$.
    
\begin{example}
If $T$ is an infinitely branching rooted tree and $S\subset V(T)$ is a connected subset containing the root of $T$, and $a,b\in V(T)-S$,
then $a\find S b$ 
holds if and only if $a$ and $b$ are in different connected components of $T-S$ 
(see \Cref{ex:forking} for an elaboration).\qed
\end{example}

Generalizing this example to infinite, nowhere dense graphs (such as infinite planar graphs), a result of Ivanov 
characterizes forking independence via a combinatorial separation property.
Here, \emph{algebraic closure} is an important closure property of sets
which however  will not be relevant in this paper,
as we will consider a stronger closure property (elementary substructure) later.
A graph $G$ is nowhere dense  if and only if the class $\set{G}$ is; equivalently, the class of its finite subgraphs is nowhere dense.

\begin{theorem}[\cite{ivanov}]\label{thm:ivanov}
    Let $G$ be a nowhere dense graph and $S\subseteq V(G)$ be an algebraically closed subset.
    For any distinct $a,b\in V(G)$,
    $a\find S b$ if and only if there is no $S$-avoiding path  connecting $a$ with $b$ in~$G$.
\end{theorem}

Our first two results
provide combinatorial characterizations of forking independence 
in monadically stable graphs. The first one 
is in terms of flip independence, establishing the link suggested in \cite{flipper-game}, and justifying the notation.
Below, $S\preceq G$ ($S$ is an elementary substructure of $G$) is a closure property 
stating for every tuple $\tup s$ of elements of $S$ and first-order formula $\phi(\tup x)$, evaluating the formula $\phi(\tup s)$ in $G$ and in the induced subgraph $G[S]$ yields the same outcome (in particular, it implies that $S$ is algebraically closed).

\begin{theorem}\label{thm:graphs-forking-flipping}
    Let $G$ be a monadically stable graph and $S\preceq G$ be an elementary substructure.
    For all $a,b\in V(G)$,
    $$a\find S b\qquad\iff\qquad a\ind S r b\quad\text{for all $r\in\N$}.$$
\end{theorem}


Our second characterization extends to all monadically stable graphs a special case (for $S\preceq G)$ of the result of Ivanov, \Cref{thm:ivanov}.
It characterizes the relation $a\find S b$ 
as connected components of a certain (symmetric) binary relation $\leftrightsquigarrow_S$
on the vertices of $G$.
We defer its definition to \Cref{sec:metric-and-disc}. By $\leftrightsquigarrow_S^\ast$ we denote its transitive, reflexive closure.

\begin{theorem}\label{thm:discrepancy}
    Let $G$ be a monadically stable graph and $S\preceq G$ be an elementary substructure. 
    For all $a,b\in V(G)$,
    $$a\find S b\qquad\iff\qquad a \leftrightsquigarrow_S^* b.$$
\end{theorem}
In the special case when $G$ is nowhere dense, $a\leftrightsquigarrow_Sb$
holds if and only if $a$ and $b$ are adjacent in $G$, and $a,b\notin S$. 
In particular, in this case, the characterization of \Cref{thm:discrepancy} coincides with that of \Cref{thm:ivanov}.

\subsection{Flipping and forking in higher arities}
 \Cref{thm:graphs-forking-flipping} above establishes a correspondence between flipping and forking in monadically stable \emph{graphs},
 partially explaining the relevance of flips in monadically stable graph classes, and characterizing forking independence combinatorially.
 Our main focus in this paper are relational structures over a finite relational language.

 \paragraph{Nowhere dense relational structures}
 As a first  step, we 
 consider the setting of \cite{nowhere-dense-structures}, of classes $\CC$ of relational structures whose Gaifman graphs form a nowhere dense graph class.
First, extending a characterization of such \emph{monotone} classes from \cite{nowhere-dense-structures} (see \Cref{thm:nowhere-dense-rel}),
we characterize all such classes as exactly those classes which are monadically dependent,
and whose class of Gaifman graphs is \emph{weakly sparse} -- avoids some biclique $K_{t,t}$ as subgraph.
\begin{restatable}{theorem}{nwdstr}\label{thm:nowhere-dense-struct}
    Let $\CC$ be a class of structures in a finite relational language and let $\Gg$ be the class of Gaifman graphs of structures in $\CC$. Then the following conditions are equivalent:
    \begin{enumerate}
        \item\label{it:nds1} $\GG$ is nowhere dense,
         \item \label{it:nds2} the monotone closure of $\CC$ is monadically stable,
        \item\label{it:nds3} $\CC$ is monadically stable and $\GG$ is weakly sparse,
        \item\label{it:nds4} $\CC$ is monadically dependent and $\GG$ is weakly sparse.
    \end{enumerate}
\end{restatable}
The implication \eqref{it:nds1}$\rightarrow$\eqref{it:nds2}
essentially follows from \cite[Proposition 5.7]{sparsity-book}
and \cite{podewski1978stable}.
We prove the implication \eqref{it:nds2}$\rightarrow$\eqref{it:nds3},
via a straightforward Ramsey argument.
The implication \eqref{it:nds3}$\rightarrow$\eqref{it:nds4}
is immediate, while the implication \eqref{it:nds4}$\rightarrow$\eqref{it:nds1}
follows from  a result of Dvo\v rak \cite{dvorakInducedSubdivisions} (see \Cref{lem:nip-weakly sparse-to-nowhere-dense}),
of which we provide a straightforward proof.


Next, we characterize forking independence (over models) in such (infinite) structures,
extending a special case of the result of Ivanov, \Cref{thm:ivanov}.

\begin{restatable}{theorem}{fndws}
    \label{thm:forking-nowhere-dense-structures}
    Let $\str M \preceq \str N$ be two structures such that the Gaifman graph $\gaif(\str N)$ of $\str N$ is nowhere dense.
    Then, for any two elements $u, v \in V(\str N) - V(\str M)$ we have $u \find {M} v$ if and only if $u$ and $v$ are in different connected components of $\gaif(\str N) - \gaif(\str M)$.
\end{restatable}

\paragraph{Definable flips of structures}
 The central definition of this paper is the following notion of  flips for arbitrary  structures, which
 is then justified by \Cref{thm:forking-for-structures,thm:flip-flatness-rel}  extending 
 \Cref{thm:graphs-forking-flipping,thm:flip-flatness,thm:fg-ranks} to monadically stable classes of relational structures.
 
 \begin{restatable}{definition}{defflips}\label{def:definable-flip-intro}\label{def:flip}
    Let $\str A$ and $\str A'$ be two structures with the same domain $A$,
    and let $S\subseteq A$.
    We say that $\str A'$ is an \emph{$S$-flip} of $\str A$ if 
    every relation of $\str A'$ is definable in $\str A$ by a quantifier-free formula with parameters from $S$,
    and vice-versa.
 \end{restatable}

 (A more concrete form of flips, also sufficient for obtaining the results in the paper, is described in \Cref{def:syntactic-flip}).

 Note that in the case of graphs, the \emph{$S$-flips} in \Cref{def:definable-flip-intro} above differ  from the \emph{$S$-definable flips} of graphs mentioned earlier -- see \Cref{ex:tree-flip} below -- although this distinction is not crucial.

\begin{example}\label{ex:tree-flip}
Coming back to the tree example discussed earlier, 
if  $T$ is a tree and $T^s$
is the graph obtained from $T$ by isolating the vertex $s$,
then $T^s$ is usually \emph{not} an $S$-flip of $T$ according to the above definition (although it is an $S$-definable flip in the sense of graph flips).
However, if $T^s$ is additionally equipped with a unary predicate $U$ marking all neighbors of $s$ in $T$,  then we have:
\iflipics
\[
    T\models E(u,v)\iff T^s\models (U(u)\land (v=s))\lor 
    (U(v)\land (u=s))\lor 
    E(u,v),        
\]
\else
\begin{multline*}
    T\models E(u,v)\iff\\ T^s\models (U(u)\land (v=s))\lor 
    (U(v)\land (u=s))\lor 
    E(u,v),        
\end{multline*}
\fi
so the edge relation of $T$ is definable in $T^s$ (equipped with the relations $E$ and  $U$) by a quantifier-free formula with parameter~$s$.
Conversely, both relations $U$ and $E$ of $T^s$ 
are definable in $T$, using quantifier-free formulas with parameter~$s$.
Thus, now $T^s$ is an $\set{s}$-flip of $T$ (and vice-versa), according to \Cref{def:definable-flip-intro}. Similarly, $T^s$ is also an $\set{s}$-flip of $\overline T$,
the edge-complement of $T$.\qed
\end{example}

\begin{example}
    This example  uses a ternary relation.
    Again, let $T$ be a rooted tree,
    and let $T_1$ be the structure 
    with domain $V(T)$, equipped with a ternary relation $R$,
    consisting of all triples $(u,v,w)$ 
    such that $u$ is the closest common ancestor of $v$ and $w$ in $T$.

    Fix a vertex $s$ of $T$, and 
    let $T_2$ be the structure 
    with domain $V(T)$,
    equipped with the following three relations:
    \begin{align*}
        R_0&=\setof{(u,v,w)}{(u,v,w)\in R, u,v,w\neq s}\\
        R_1&=\setof{(v,w)}{(s,v,w)\in R}\\
        R_2&=\setof{(u,w)}{(u,s,w)\in R}.
    \end{align*}
    Then $T_1$ and $T_2$ are $\set{s}$-flips of each other; moreover, $s$ is isolated in the Gaifman graph of $T_2$.\qed
\end{example}

\paragraph{Flip independence}
Following \cite{flipper-game}, we introduce the notion of flip independence in relational structures.
\begin{definition}    
    Fix $r\in\N$, a relational structure $\str A$, and set $S$ of its elements.
    For $a,b\in V(\str A)$, we write 
    $$a\ind S r b$$ to 
    denote that there is some $S$-flip $\str A'$ of $\str A$ such that there is no path of length at most $r$ connecting $a$ with $b$ 
    in the Gaifman graph of $\str A'$,
    or $a=b\in S$.
\end{definition}

Our main result is the following combinatorial characterization of forking independence (over models) in monadically stable relational structures,  extending \Cref{thm:graphs-forking-flipping}.
\begin{restatable}{theorem}{forkingforstructures}
    \label{thm:forking-for-structures}
    Let $\str M \preceq \str N$ be two monadically stable structures in a  relational language.
    For all $a,b\in N$,
    $$a\find M b\qquad\iff\qquad a\ind M r b\quad\text{for all $r\in\N$}.$$
\end{restatable}

Moreover, by a simple application of a result of Braunfeld and Laskowski \cite[Theorem 1.1]{braunfeld2021characterizations}, we obtain a converse result assuming that $G$ is stable.
\begin{restatable}{proposition}{forkingforstructuresconverse}
    \label{prop:forking-for-structures-converse}
    Let $\str M$ be a monadically stable structure and assume that for every elementary extension $\str M \preceq \str N$, every elementary substructure $\str S \preceq \str N$, and all $a,b\in V(\str N)$
    $$a\find S b\qquad\iff\qquad a\ind S r b\quad\text{for all $r\in\N$}.$$
    Then $\str M$ is monadically stable.
\end{restatable}

\subsection{Applications: flip-flatness and separation game}\label{sec:intro-app}
As two applications of \Cref{thm:forking-for-structures}, we prove the following two combinatorial characterizations of monadically stable classes of relational structures,
extending \Cref{thm:flip-flatness,thm:fg-ranks}.

The following notion is 
 an analogue of the condition in \Cref{thm:flip-flatness} from \cite{flip-flatness}.

\begin{definition}\label{def:flip-flat}
    Let $\CC$ be a class of structures in a finite relational language.
    The class $\CC$ is \emph{flip-flat} if for every $r\in\N$
    there is some $k\in\N$ and an unbounded function $U\from\N\to\N$ with the following property.
        For every $\str A\in\CC$ and set $X\subseteq V(\str A)$, there is a 
        set $S\subseteq V(\str A)$ with $|S|\le k$, an $S$-flip $\str A'$ of $\str A$,
        and a set $Y\subseteq X$ with $|Y|\ge U(|X|)$,
        whose elements have pairwise distance greater than $r$ in the Gaifman graph of $\str A'$.
\end{definition}
    
Fix $r\in\N$ and a structure $\str A$.
For two sets $U,S\subseteq V(\str A)$,
define the \emph{separation rank} of $U$ over $S$, denoted 
$\textup{srk}_r^{\str A}(U/S)$,
by $\textup{srk}_r^{\str A}(U/S)\coloneqq 0$
if $U\subseteq S$,
and otherwise,
$$\textup{srk}_r^{\str A}(U/S)\coloneqq1+\min_{s\in \str A}\ \max_{v\in U}\ 
\textup{srk}_r^{\str A}(U\cap B^r_{Ss}(v)/Ss),$$
where $Ss=S\cup\set s$ and $B^r_{Ss}(v)=\setof{u\in \str A}{v\nind {Ss} r u}$.
Note that for all $U,S\subseteq V(\str A)$
we have that 
$$\textup{srk}_r^{\str A}(U/S)\le 
\textup{srk}_r^{\str A}(V(\str A)/\emptyset)\eqqcolon \textup{srk}_r(\str A).$$
The parameter $\textup{srk}_r(\str A)$
can be explained in terms of a two player game on $\str A$,
see \Cref{sec:apps}.
In the case of graph classes, the condition that $\textup{srk}_r(\CC)<\infty$ holds for all $r\in\N$ 
 implies the condition in \Cref{thm:fg-ranks}  (see  \cite[Section 3]{flipper-game-arxiv}).

\begin{theorem}\label{thm:flip-flatness-rel}
    Let $\CC$ be a class of structures in a finite relational language. The following conditions are equivalent:
    \begin{enumerate}
        \item $\CC$ is monadically stable,
        \item $\CC$ is flip-flat,
        \item for every $r\in\N$ there is some $k\in\N$ such that ${\textup{srk}_r(\str A)<k}$ holds for all $\str A\in\CC$.
    \end{enumerate}
\end{theorem}
As was the case with \Cref{thm:fg-ranks}, we expect that 
the last condition in \Cref{thm:flip-flatness-rel} will be important in resolving the 
tractability of monadically stable classes of relational structures.

\subsection{Relationship to other work}
Baldwin and Shelah \cite{baldwin-shelah} introduced the notion of monadic stability, and studied this concept from a model-theoretic perspective. Some of our results resemble 
(or reprove) some of their results in a more explicit, combinatorial setting; for instance, the last condition in \Cref{thm:flip-flatness-rel} is closely related (and a finitary analogue) of the decompositions of models from \cite{baldwin-shelah}, while \Cref{thm:forking-for-structures} implies that the forking relation on singletons defines an equivalence relation (c.f. \cite[Lemma 4.2.6]{baldwin-shelah}).

Braunfeld and Laskowski \cite{braunfeld2022existential} provide a more precise model-theoretic understanding of monadically stable and monadically dependent classes;
in particular, they study the existential theories of such classes.

Gajarsk\'y, M\"ahlmann, McCarty, Ohlmann, Pilipczuk, Przybyszewski, Siebertz, Soko{\l}owski, and Toruńczyk
   \cite{flipper-game} introduced the notion of flip independence, and proved a fundamental property (called \emph{$r$-separability}) of monadically stable graph classes, from which they then derived combinatorial properties, related to the flipper game, flip-flatness, and a characterization via obstructions.
We follow the approach of \cite{flipper-game}, and extend their results in two ways.
First, by a refined analysis of flip independence, we obtain 
a combinatorial characterization of forking independence in monadically stable graphs. Second, we 
extend their approach to monadically stable relational structures. To this end, we introduce a notion of flips in relational structures. Lifting the results from graphs to relational structures makes the analysis technically more challenging. While flips of graphs remain graphs, it seems essential that in order to study flips of relational structures, 
more relations need to be introduced.
In the case of graphs, a rather explicit description of the obstructions to monadic stability follows from the constructions (see \cite[Section 6]{flipper-game-arxiv}), which can then be used for proving algorithmic hardness results (see \cite{ms-mc2}). In the case of relational structures, we believe that a finer analysis of the proof from \Cref{sec:flip-fork} might lead to similar results; however, the analysis becomes much more cumbersome than in graphs, due to the extra added relations. 

Braunfeld, Dawar, Eleftheriadis, and Papadopoulos
 \cite{nowhere-dense-structures-icalp} 
study monotone, monadically stable classes of relational structures, and characterizes them. We provide a straightforward proof of a more general result
(see \Cref{thm:nowhere-dense-struct}).
However, we do not obtain algorithmic hardness results; this would require additional Ramsey-based arguments.

Ivanov \cite{ivanov} characterizes forking independence in nowhere dense graphs. Our results generalize his result in two ways: by characterizing (in \Cref{thm:discrepancy}) forking independence in monadically stable graphs (although only over models), and by characterizing (in \Cref{thm:forking-nowhere-dense-structures}) forking independence (over models) in relational structures whose Gaifman graphs are nowhere dense.

\paragraph*{Paper outline}
In \Cref{sec:prelims} we introduce the required basic notions from logic, model theory, and Ramsey theory.
In \Cref{sec:flip-ind-graphs}, 
we characterize forking independence in monadically stable graphs. Those results are later generalized to relational structures in \Cref{sec:flip-fork}. In the graph setting, we obtain a finer analysis,
while in \Cref{sec:flip-fork}, the proofs become more technical.
In \Cref{sec:nd-struc} we study relational structures whose Gaifman graphs are nowhere dense. In \Cref{sec:flip-fork} we prove our main result, characterizing forking independence in monadically stable structures. In \Cref{sec:apps}, we derive several characterizations of monadically stable classes, some of which are purely combinatorial. 

Due to space constraints, some proofs are relegated to the appendix.
\section{Preliminaries}\label{sec:prelims}
We denote $\N\coloneqq\set{0,1,2,\ldots}$, and for $n\in\N$ 
we denote $[n]\coloneqq\set{1,2,\ldots,n}$.

\subsection{Structures, logic}
In this paper, we only consider logical structures in a language $\cal L$
consisting of finitely many relation symbols and (possibly infinitely many) constant symbols.
Structures are denoted with boldface letters $\str A,\str B,\str M,\str N$, etc.
and their domains are denoted with the corresponding italic letters $A,B,M,N$, etc. We also write $V(\str A)$ to denote the domain of $\str A$. 
We write $R_{\str A}$ to denote the interpretation of a relation symbol $R\in\cal L$ 
in the structure $\str A$.

Graphs are undirected, without self-loops,
that is, a graph $G$ consists of a set of vertices $V(G)$, and a set of edges $E(G)\subseteq {V(G)\choose 2}$.
A graph $G=(V,E)$ is viewed as a relational structure equipped with one binary relation $E$ denoting adjacency.

For two structures $\str A$ and $\str B$ over the same relational language, and with $V(\str A)\subseteq V(\str B)$,
we say that $\str A$ is an \emph{induced substructure}  of $\str B$ if 
for every $k$ and $k$-ary relation  symbol $R$, and for every $k$-tuple $\tup a\in A^k$,
we have 
$$\tup a\in R_{\str A}\quad \iff \quad \tup a\in R_{\str B}.$$
A class $\CC$ of relational structures is \emph{hereditary} if for every $\str B\in \CC$, if $\str A$ is an induced substructure of $\str B$, then $\str A\in\CC$.

A \emph{reduct} of a structure $\str A$ is a structure $\str B$ 
obtained from $\str A$ by forgetting some relations and constants.
A \emph{lift} of $\str A$ is a structure $\str A^+$ such that $\str A$ is a reduct of $\str A^+$. If all the relations of $\str A^+$ which are forgotten in $\str A$ are unary, we say that $\str A^+$ is a \emph{monadic lift} of $\str A$.

In this paper, we only consider formulas of first-order logic.
An \emph{atomic formula}
is a formula of the form $R(\tau_1,\ldots,\tau_k)$,
where $R$ is a $k$-ary relation symbol in $\cal L\cup\set{=}$,
and each $\tau_i$ is either a variable, or a constant symbol in $\cal L$. 
For a formula $\phi$ and set $\bar x$ of variables, we sometimes write $\phi(\bar x)$ to indicate that $\bar x$ is the set of free variables of $\phi$.
A \emph{partitioned formula} $\phi(\bar x;\bar y)$ is a formula $\phi$ whose set of free variables is partitioned into two distinguished sets $\bar x$ and~$\bar y$.

If $\str M$ is an $\cal L$-structure and $S\subseteq V(\str M)$,
then 
a \emph{formula with parameters from $S$} is a formula $\phi(\bar x)$ 
in the language $\cal L\cup S$,
where each element $s\in S$ is viewed as a constant symbol,
interpreted as $s$.

\subsection{Gaifman graph}
The \emph{Gaifman graph} of a structure $\str A$,
denoted $\gaif(\str A)$, is the graph 
with vertices $V(\str A)$ and edges $ab$ such that 
there is some relation $R$ of $\str A$ such that 
$a$ and $b$ occur simultaneously in some tuple in $R_{\str A}$.
For a structure $\str A$ and elements $u,w\in V(\str A)$,
let $\dist_{\str A}(u,w)$ denote the distance between $u$ and $v$ in $\gaif(\str A)$ (that is, the length of the shortest path between $u$ and $w$, or $\infty$ if no path exists).
For two sets or tuples $U,W$ of elements of $\str A$,
 let $\dist_{\str A}(U,W)$ denote the minimum distance 
$\dist_{\str A}(u,w)$, for $u\in U,w\in W$.
For an element $a \in V(\str A)$ and $r\in\N$, by $B^r_{\str A}(a)$ denote the ball of radius $r$ around $a$ in $\gaif(\str A)$.
If the structure $\str A$ is clear from the context,
we omit it from the subscript in $\dist_{\str A}(\cdot,\cdot)$ and $B^r_{\str A}(\cdot)$.

The following result is a consequence of Gaifman's locality theorem \cite{gaifman}.

\begin{restatable}{fact}{foloc}
     \label{fact:fo-locality}
    For every first-order formula $\phi(\bar x; \bar y)$ of quantifier rank $q$ there is some $k\in \N$
    and formulas $\alpha_{1}(\bar x),\ldots,\alpha_{k}(\bar x)$ and $\beta_1(\bar y),\ldots,\beta_k(\bar y)$ 
    such that for every structure $\str M$ and tuples $\bar a\in M^{\tup x},\bar b\in M^{\tup y}$, if $\dist_{\str M}(\bar a,\bar b)>7^q$, then 
    \[
        \str M \models \phi(\bar a; \bar b) \iff \str M \models \bigvee_{j=1}^k\alpha_j(\bar a)\land \beta_j(\bar b).
    \]
\end{restatable}
We only use the following.

\begin{corollary}\label{cor:fo-locality}
    For every first-order formula $\phi(\bar x; \bar y)$ of quantifier rank $q$,
 structure $\str M$ and tuple $\bar b\in M^{\tup y}$,
    there is a first-order formula $\alpha(\bar x)$ with the following property.
    For every tuple $\tup a\in M^{\tup x}$
     with $\dist_{\str M}(\bar a,\tup b)>7^q$, we have: 
    \[
        \str M \models \phi(\bar a; \bar b) \iff \str M \models \alpha(\tup a).
    \]
\end{corollary}

\subsection{Types}
Let $\str M$ be a structure, $\phi(\tup x;\tup y)$ be a partitioned formula, and let $B \subseteq M$ a set of elements.
The \emph{$\phi$-type} of a tuple $\tup a\in M^{\tup x}$ \emph{over $B$}, denoted $\tp^\phi(\tup a/B)$ is the set 
of all formulas $\phi(\tup x;\tup b)$ 
with parameters $\tup b$ in $B$
such that 
$\str M \models \phi(\tup a;\tup b).$
For $A\subseteq  M^{\tup x}$, we write $$\Types\phi(A/B)$$
to denote the set of $\phi$-types of tuples $\tup a\in A$.

The  \emph{type} of $\tup a$ over $B$,
denoted $\tp(\tup a/B)$,
is the set of 
all  formulas $\phi(\tup x;\tup b)$,
where $\phi(\tup x;\tup y)$ is a partitioned formula and $\tup b$ is a tuple of elements of $B$,
such that $\str M \models \phi(\tup a;\tup b).$

The \emph{atomic type} of $\tup a$ over $B$,
denoted $\atp(\tup a/B)$, is defined analogously, with $\phi$ ranging only over atomic formulas.
The \emph{atomic type} of $\tup a$ is its 
atomic type over $\emptyset$.

 \subsection{Model theory and stability}\label{sec:mt-prelims}
For two structures $\str A,\str B$ with $V(\str A)\subseteq V(\str B)$, we say that $\str A$ is an \emph{elementary substructure} of $\str B$, and write $\str A\preceq \str B$,
if for every first-order formula $\phi(\bar x)$
and any tuple $\tup a\in A^{\tup x}$,
we have 
$$\str A\models\phi(\tup a)\quad\iff\quad \str B\models\phi(\tup a).$$

Let $\CC$ be a class of structures in some language $\cal L$.
A partitioned formula $\phi(\tup x;\tup y)$ 
is \emph{stable} in $\CC$
if there is some $n\in\N$ 
such that for every $\str A\in\CC$,
there are no tuples 
$\tup a_1,\ldots,\tup a_n\in A^{\tup x}$
and $\tup b_1,\ldots,\tup b_n\in A^{\tup y}$ 
such that
$$\str A\models\phi(\tup a_i;\tup b_j)\quad\iff \quad i\le j\qquad \text{ for all $i,j\in \set{1,\ldots,n}$}.$$

A class $\CC$ is \emph{stable}, resp. \emph{atomic-stable},
if every partitioned formula, resp. atomic formula, is stable in $\CC$.
A class $\CC$ is \emph{monadically stable} 
if for every class $\CC^+$ such that 
every structure $\str A^+\in\CC^+$ is a monadic lift of some $\str A\in\CC$, the class $\CC^+$ is stable.
For hereditary classes, monadic stability is equivalent to the definition from the introduction \cite{braunfeld2022existential}, see \Cref{thm:bl-ms} below.
Examples of monadically stable classes 
coming from graph theory include 
all nowhere dense classes (see below),
as well as any class obtained from a nowhere dense graph class 
$\CC$ by means of a first-order formula $\phi(x_1,\ldots,x_k)$
(namely, for each graph $G\in\CC$ construct a structure $\phi(G)$ 
with domain $V(G)$ and relation $R$ consisting of those $k$-tuples 
$(a_1,\ldots,a_k)$ with $G\models \phi(a_1,\ldots,a_k)$).

A class $\CC$ is \emph{dependent},
resp. \emph{existentially dependent}
if for every partitioned formula $\phi(\tup x;\tup y)$ 
(resp. existential partitioned formula)
there is some $n\in\N$ 
such that 
for every $\str A\in\CC$, there are no tuples 
$(\tup a_i:i\in [n])$ and $(\tup b_J:J\subseteq [n])$
such that 
$$\str A\models\phi(\tup a_i;\tup b_J)\quad\iff \quad i\in J\qquad \text{ for all $i\in [n],J\subseteq [n]$}.$$
A class $\CC$ is \emph{monadically dependent},
resp. \emph{existentially monadically dependent},
if for every class $\CC^+$ such that 
every structure $\str A^+\in\CC^+$ is a monadic lift of some $\str A\in\CC$, the class $\CC^+$ is dependent, resp. existentially dependent.
Clearly, every stable class is dependent,
and every monadically stable class is monadically dependent.
Monadically dependent graph classes include 
all monadically stable graph classes (such as nowhere dense graph classes),
every class of bounded clique-width, and more generally, every class of bounded twin-width \cite{tww1} or of bounded flip-width \cite{flip-width}.
Note that the adjective \emph{dependent} is sometimes replaced with \emph{NIP} in model theory.
A combinatorial characterization of monadically dependent graph classes is given in \cite{flip-breakability}.

\begin{theorem}[\cite{braunfeld2022existential}]\label{thm:bl-ms}
    The following are equivalent for a hereditary class of  structures in a finite relational language:
    \begin{enumerate}
        \item $\CC$ is stable,
        \item $\CC$ is monadically stable,
        \item $\CC$ is monadically dependent and atomic-stable.
    \end{enumerate}
\end{theorem}

The following lemma is immediate.
\begin{lemma}\label{lem:to-atomic-stable-and-emNIP}
    Every monadically stable class of relational structures 
    is atomic-stable and existentially monadically dependent.
\end{lemma}

For a class $\CC$ of structures, by $\overline \CC$ 
we denote the \emph{elementary closure} of $\CC$,
which is the class of all structures $\str M$ 
which satisfy all the sentences $\phi$ that hold 
in all structures in $\CC$. Note that if $\CC$ contains 
structures of arbitrarily large size, then $\overline\CC$ contains infinite structures (even of arbitrarily large cardinality).

It is not difficult to prove that 
atomic-stability and existential monadic dependence (and similarly existential monadic stability) are preserved by taking the elementary closure.
In some of our proofs, this is why we consider those notions, rather than monadic stability, for which this is is unclear (but ultimately follows from the equivalences in \Cref{sec:apps}, or the results of \cite{braunfeld2022existential}).
\begin{restatable}{lemma}{closure}\label{lem:closure}
    Let $\CC$ be an atomic-stable class, resp. an existentially monadically dependent class of structures.
    Then $\overline \CC$ is atomic-stable, resp. existentially monadically dependent.
\end{restatable}
We provide a proof of \Cref{lem:closure} in Appendix \ref{app:prelims}.

\subsection{Forking independence and definability of types}

Forking independence is a fundamental tool in the study of stable classes.
Instead of recalling the original definition of forking independence, we use the following characterization of forking independence over models in stable theories in terms of so-called \emph{finite satisfiability} (c.f. Definition 8.1.1 and Corollary 8.3.7 in \cite{tent-ziegler}).

\begin{definition}[Finite satisfiability]
    Let $\str N$ be a structure, $A \subseteq B\subseteq N$, and $\tup v\in N^k$.
    We say that \emph{$\tp(\tup v/B)$ is finitely satisfiable in $A$}
    if for every formula $\phi(x_1,\ldots,x_k)\in \tp(\tup v/B)$ (with parameters from $B$),
    there is some $\tup w\in A$ such that $\phi(x_1,\ldots,x_k)\in \tp(\tup w/B)$.
\end{definition}
\begin{definition}[Forking over models]\label{def:forking}
    Let $\str M \preceq \str N$ be two stable structures and let $B\subseteq N$.
    A set 
    $A\subseteq N$ is \emph{forking independent} of $B$ over $M$, denoted $A \find {M} B$,
    if for every tuple $\tup a\in A^k$, $\tp(\tup a/M\cup B)$ is finitely satisfiable in $M$.
\end{definition}

Observe that if $\str M\preceq \str N$
then $N\find{M}M$.
In other words, $\tp(\tup a/M)$ is finitely satisfiable in $M$, 
for all $\tup a\in N^k$.
Indeed, suppose $\phi(\tup x;\tup m)$ is a formula with parameters $\tup m\in M^\ell$
with $\str N\models \phi(\tup a;\tup m)$.
Then $\str N\models\exists \tup x.\phi(\tup x;\tup m)$,
and so $\str M\models\exists \tup x.\phi(\tup x;\tup m)$ since 
 $\str M\preceq \str N$.
Therefore, there is some $\tup a'\in M^k$ with $\str N\models\phi(\tup a';\tup m)$.

In the case of monadically stable structures Baldwin and Shelah proved that for any sets $A, B, C$ the independence $A \find {C} B$ holds if and only if $a \find C b$ for all $a \in A$ and $b \in B$ (see \cite[Section 4.2]{baldwin-shelah}).
Therefore, when characterizing forking independence over models in monadically stable structures, we can restrict ourselves to the case of singletons.
We remark that our work can be easily generalized to the case of tuples instead of singletons, thus reproving the result of \cite{baldwin-shelah} in the case of forking over models.

\begin{example}\label{ex:forking}
    Let $T$ be an infinite rooted tree (in the graph-theoretic sense), with infinite branching at every node.
    Let $S\subset T$ be any subtree of $T$ 
which contains the root of $T$, and also has infinite branching at every node. One can show that $S\preceq T$.

\sz{draw pic}
Then, for any $a,b\in T-S$ we 
have: $$a\find S b\iff \text{there is no $S$-avoiding path from $a$ to $b$ in $T$}.$$
We first prove the left-to-right implication, by contrapositive. Suppose that there is some $S$-avoiding path from $a$ to $b$; let $r$ be its length. Let $s\in S$ be the closest ancestor of $b$ in $S$. Let $\phi(x)$ be the formula 
with parameters $s$ and $b$ expressing that there is a path of length $r$ from 
$x$ to $b$ which avoids $s$.
Then $T\models\phi(a)$, so $\phi(x)\in \tp(a/S\cup \set{b})$.
Suppose $a'\in S$ 
is such that $\phi(x)\in \tp(a'/S\cup \set{b})$,
that is, $T\models \phi(a')$.
This is a contradiction, since every path from $a'\in S$ 
to $b\in T-S$ must pass through $s$.
Therefore, $\tp(a/S\cup \set{b})$ is not finitely satisfiable in $S$, which proves 
that $a\nfind S b$.

\smallskip
Now, suppose that there is no $S$-avoiding path from $a$ to $b$ in $T$. Denote  the closest ancestor of $a$ in $S$ by $s$.

Let $\phi(x)\in \tp(a/S\cup\set{b})$
and let $S_0\subset S$ be the (finite) set of those elements of $S$ which occur as parameters in $\phi$.
We show that there is some $a'\in S$ with $\phi(x)\in \tp(a'/S\cup\set{b})$.\sz{simplify?}



Consider the subtree $T_s$ of $T$ consisting of $s$ and all descendants of $s$.
Then $a\in T_s$, $b\notin T_s$, and $T_s\cap S$ 
forms an infinite branching subtree of $T_s$.
Pick an automorphism $\alpha$ of $T_s$ 
which is the identity on  $S_0\cap T_s$ (and all their finitely many ancestors in $T_s$) and maps $a$ to some element $a'\in T_s\cap S$.
The automorphism $\alpha$ extends to an automorphism $\alpha'$ of $T$
which is the identity on $(T-T_s)\cup S_0\cup\set{s}$,
and maps $a$ to $a'\in S$.
In particular, $\alpha'$ fixes $S_0\cup\set{b}$, as $b\in T-T_s$. This implies that  $\phi(x)\in \tp(a'/S\cup\set{b})$,
so $\tp(a/S\cup\set{b})$ is finitely satisfiable in $S$,
proving $a\find S b$.\qed
\end{example}

The following result is a fundamental fact about stable formulas (see e.g. \cite[Lemma~2.2~(i)]{pillay}).

\begin{theorem}[Definability of types]
    \label{thm:definability-of-types}
    Let $\str M\preceq \str N$ be two structures and $\phi(\tup x; \tup y)$ be a formula that is stable in $\str M$.
    For every $\tup n \in \str N^{\tup x}$ there is some formula $\psi_{\tup n}(\tup y)$ with parameters from $\str M$, which is a positive boolean combination of formulas of the form $\phi(\tup m; \tup y)$ for $\tup m\in\str M^{\tup x}$, such that:
    \begin{align}\label{eq:def-types}
        \str N\models \phi(\tup n; \tup a)\ \iff\     \str M\models \psi_{\tup n}(\tup a)\qquad \text{ for all }\tup a\in\str M^{\tup y}.
    \end{align}
    Any such formula $\psi_{\tup n}(\tup y)$ is called a \emph{definition} of the $\phi$-type of $\tup n$ over  $\str M$. 
    Moreover, the size of the formula $\psi_{\tup n}(\tup y)$ can be bounded by some constant depending only on $\str N$ and $\phi$.
\end{theorem}
Note  that there is a unique first-order formula $\psi_{\tup n}(\tup y)$ satisfying \eqref{eq:def-types}, up to equivalence in $\str N$:
If $\psi_{\bar n}(\bar y)$ and $\psi'_{\bar n}(\bar y)$ are two such formulas,
then by \eqref{eq:def-types}, for all $\bar a\in \str M^{\tup y}$ we have 
$\str N\models \psi_{\tup n}(\tup a)\iff \psi_{\tup n}'(\tup a)$.
Since $\str M\preceq \str N$, this implies that 
for all $\tup a\in \str N^{\tup y}$ we have 
$\str N\models \psi_{\tup n}(\tup a)\iff \psi_{\tup n}'(\tup a)$.


We also state a version of the Harrington's Lemma (c.f. \cite[Lemma 8.3.4]{tent-ziegler}) about  definitions of stable formulas.
\begin{lemma}[Harrington's lemma]
    \label{lem:harrington}
    Let $\str M\preceq \str N$ be two structures and $\phi(\tup x; \tup y)$ be a formula that is stable in $\str M$.
    For any $\tup n \in \str N^{\tup x}, \tup n' \in \str N^{\tup y}$, any definition $\psi_{\tup n}(\tup y)$ of the $\phi(\tup x;\tup y)$-type
    of $\tup n$ over $\str M$, and any definition $\psi_{\tup n'}(\tup x)$ of the $\phi(\tup y;\tup x)$-type of $\tup n'$ over $\str M$ we have $${\str N \models \psi_{\tup n}(\tup n') \iff \psi_{\tup n'}(\tup n)}.$$
\end{lemma}
\begin{corollary}
    \label{lem:complete-anticomplete-tuples}Let $\str M \preceq \str N$ be two structures 
and let $\phi(\tup x; \tup y)$ be a formula  that is stable in $\str M$.
    There are no two tuples $\tup s \in N^{\tup x}, \tup s' \in N^{\tup y} $ such that for every $\tup m \in M^{\tup y}$ we have $\str N \models \phi(\tup s; \tup m)$ and for every $\tup m' \in M^{\tup x}$ we have $\str N \not \models \phi(\tup m'; \tup s')$.
\end{corollary}


The following lemma is reminiscent of the classic notion of Morley sequences from model theory (see e.g. \cite[Definition 2.27]{pillay}), which was mentioned in the context of graphs in \cite[Lemma 5.4]{flipper-game}.
\begin{lemma}
    \label{lem:morley-sequence}
    Let $\str M \preceq \str N$ be structures in the language consisting of one relation $R$ of arity $r$, let $k \in \N$ be a fixed constant, let $\tup n \in \str N^{\tup x}$, and let $A \subseteq \str M$ be a finite set.
    There is an infinite sequence $\tup n_0, \tup n_1, \tup n_2, \ldots \in \str M^{\tup x}$ such that
    the following holds for $i=0,1,2,\ldots$:
    \begin{itemize}
        \item $\tup n_i$ and $\tup n$ have equal atomic types over $A \cup \tup n_0 \cup \ldots \cup \tup n_{i - 1}$; and
        \item the atomic types of the tuples $\tup n_{i_1}\ldots\tup n_{i_k}$ are the same for all $i_1 < \ldots < i_k$.
    \end{itemize}
    Moreover, if every atomic formula is stable in $\str N$,  the sequence $(\tup n_i)_{i \in \N}$ may be chosen so that the atomic types of $\tup n_{i_1}\ldots\tup n_{i_k}$ are the same for all pairwise distinct $i_1, \ldots, i_k$.
\end{lemma}

\subsection{Ramsey theory}
We use the following consequence of Ramsey's theorem.
For a totally ordered set $(I,\le)$
and tuple $\tup a\in I^k$, the \emph{order type} of $\tup a$,
denoted $\otp(\tup a)$, is the atomic type of the tuple $\tup a$ in the structure $(I,\le)$.

\begin{lemma}[{see \cite[Lemma 4.4]{flipbreakability-arxiv}}]\label{thm:bi-ramsey-tuples}
    Fix $k,\ell\in\N$.
    There is an unbounded function $U_{k,\ell}\from\N\to\N$
    such that for every two sets $I,J$ with $|I|=|J|$
    and function $f\from I^\ell\times J^\ell\to[k]$,
    there are subsets $I'\subset I$ and $J'\subset J$
    with $|I'|=|J'|\ge U_{k,\ell}(|I|)$,
    such that for all 
    $\tup a\in (I')^\ell$ and $\tup b\in (J')^\ell$,
    the value $f(\tup a,\tup b)$ depends only on $\otp(\tup a)$ and $\otp(\tup b)$, i.e. there is a function $g$ such that 
    $f(\tup a,\tup b)=g(\otp(\tup a),\otp(\tup b))$
for all $\tup a\in (I')^\ell$ and ${\tup b\in (J')^\ell}$.
\end{lemma}

We also use the following result.
\begin{lemma}[see \cite{GRAVIER2004719}]\label{thm:folklore}
    Let $G=(A,B,E)$ be a bipartite graph
    with $\Types E(A/B)$ infinite.
    There is a relation ${\sim}\in\set{=,\neq,\le,\ge}$ and pairwise distinct elements $a_1,a_2,\ldots \in A,b_1,b_2,\ldots\in B$
      such that 
$$G\models E(a_i,b_j)\iff i\sim j\qquad\text{for }i,j\in\N.$$
\end{lemma}

\subsection{Sparsity}
A subgraph of a graph $G$ is a graph obtained from $G$ by removing vertices and edges arbitrarily.
A graph class is \emph{monotone} if it is closed under taking subgraphs.

A graph class $\CC$ is \emph{nowhere dense}
if for every $r\in\N$ there is some $n\in\N$ such that 
the $r$-subdivision of $K_n$ (obtained by replacing each edge of $K_n$ with a path of length $r-1$)
is not a subgraph of any graph $G\in\CC$.

The following characterization of nowhere dense graph classes is known.
A graph class $\CC$
is \emph{weakly sparse} if some biclique $K_{t,t}$ 
is not a subgraph of any graph in~$\CC$.

Nowhere dense graph classes include: the class of planar graphs, any class which excludes a fixed graph as a minor, and any class of graphs of maximum degree bounded by a fixed constant. On the other hand, the 
class of cliques, or the class of bicliques, is not nowhere dense, but is monadically stable.

\begin{theorem}[Consequence of \cite{podewski1978stable,AdlerA14} and \cite{dvorakInducedSubdivisions}]\label{thm:nowhere-dense-g}
    The following conditions are equivalent for a graph class $\CC$:
    \begin{enumerate}
        \item\label{it:nds1-g} $\CC$ is nowhere dense,
         \item \label{it:nds2-g} the monotone closure of $\CC$ is monadically stable,
        \item\label{it:nds3-g} $\CC$ is monadically stable and is weakly sparse,
        \item \label{it:nds4-g} $\CC$ is monadically dependent and is weakly sparse.
    \end{enumerate}
\end{theorem}
For the implication \eqref{it:nds1-g}$\rightarrow$\eqref{it:nds2-g},
it is enough to show that every nowhere dense class $\CC$ is monadically stable. This is the result of \cite{podewski1978stable}, see also \cite{AdlerA14}.
The implication \eqref{it:nds2-g}$\rightarrow$\eqref{it:nds3-g}
is immediate, since if $\CC$ is not weakly sparse then its monotone closure contains all bipartite graphs, and is therefore not monadically stable.
The implication \eqref{it:nds3-g}$\rightarrow$\eqref{it:nds4-g} is immediate. 
The implication \eqref{it:nds4-g}$\rightarrow$\eqref{it:nds1-g} follows from the following lemma, which can be derived from a result of Dvo\v rak \cite{dvorakInducedSubdivisions}. We provide a more direct, self-contained proof in Appendix \ref{app:prelims}. 
    
    \begin{restatable}{lemma}{dvorak}
        \label{lem:nip-weakly sparse-to-nowhere-dense}
        Let $\CC$ be a monadically dependent, weakly sparse graph class.
        Then $\CC$ is nowhere dense.
    \end{restatable}
    
    

    \paragraph*{Convention}
    We say that 
    a (usually infinite) graph or structure $\str M$ is 
    nowhere dense or monadically stable,
    if the class $\set{\str M}$ has the corresponding property.
    
\section{Forking independence over models in monadically stable graphs}\label{sec:flip-ind-graphs}
The goal of this section is to provide a combinatorial characterization 
of forking independence in monadically stable graphs over models,
in terms of \emph{flip independence}.

We first recall the relevant definitions and lemmas from \cite{flipper-game}, and then state the main result of the section.

\begin{definition}\label{def:flips-graphs}
    Let $G$ be a graph and $\cal P$ be a partition of its vertices.
    A \emph{$\cal P$-flip} of $G$ is a graph $G'$ with the same vertex set as $G$,
    and such that for all $A,B\in\cal P$
    one of the following cases holds:
    \begin{itemize}
        \item $G\models E(a,b)\iff G'\models E(a,b)$ holds for all distinct $a\in A,b\in B$, or
         \item $G\models E(a,b)\iff G'\not\models E(a,b)$ holds for all distinct $a\in A,b\in B$.
    \end{itemize}
    If $S\subseteq V(G)$ is a finite vertex set, then an \emph{$S$-definable flip} of $G$ 
    is a $\cal P$-flip of $G$, where $\cal P$ is the partition of $V(G)$ into atomic types over~$S$
    (that is, the parts of $\cal P$ are sets of vertices with equal atomic types over $S$).
    \sz{inconsistent with intro}
\end{definition}

\begin{definition}
    Let $G$ be a graph and $A, B, C \subseteq V(G)$ be three sets of its vertices.
    We say that $A$ is \emph{radius-$r$ flip independent} of $C$ over $B$ and write $A \ind Br C$ if there is a finite set $B' \subseteq B$ and a $B'$-definable flip $G'$ of $G$ such that 
    $\dist_{G'}(a,c)>r$ holds for all distinct
     $a \in A$ and $c \in C$.
\end{definition}

The following is the main result of \Cref{sec:flip-ind-graphs},
and the first result of this paper:
a~combinatorial characterization of forking independence over models in monadically stable graphs in terms of flip independence.


\begingroup
\renewcommand\thetheorem{\ref{thm:graphs-forking-flipping}} 
\begin{restatable}{theorem}{forkinggraphs}
    \label{thm:forking-characterization-graphs}
    Let $G \preceq H$ be monadically stable graphs and let $u, v \in H \setminus G$.
    Then $u \find G v$ if and only if $u \ind Gr v$ for every $r \in \N$.
  \end{restatable}
  \endgroup

  The following lemma establishes a kind of transitivity of the flip dependence relation. It can be viewed as the triangle inequality for a suitable metric; see the discussion following \Cref{lem:structures-flip-dependence-transitivity}.
  \begin{lemma}
      \label{lem:flip-dependence-transitivity}
      Let $H$ be a graph and $S\subseteq V(H)$.
      If  $r, q \in \N$ are positive integers, then $u \nind S{r} v$ and  $v \nind S{q} w$ imply $u \nind S{r + q} w$.
  \end{lemma}
  \begin{proof}
  
      Towards a contradiction, assume that $u \ind G{r + q} w$.
      Therefore, there is an $S$-definable flip $H'$ of $H$ such that $\dist_{H'}(u, w) > r + q$.
      As $u \nind S{r} v$ we have $\dist_{H'}(u, v) \le r$, so by the triangle inequality $\dist_{H'}(v, w) > q$.
      This is a contradiction with $v \nind S{q} w$.
  \end{proof}
  
  \Cref{thm:forking-characterization-graphs} implies  that if $G\preceq H$ are monadically stable graphs, the non-independence relation $u \nfind G v$ on singletons (from  $H$) is equivalent to the existence of some $r\in \N$ such that $u\nind Gr v$. 
  By \Cref{lem:flip-dependence-transitivity},
  it follows that the relation $\nfind G$ is reflexive, symmetric and transitive, so it is an equivalence relation on $H$ (this also follows from the results of \cite{baldwin-shelah}).

  \medskip
  \Cref{thm:forking-characterization-graphs} will be later generalized 
  to relational structures, in \Cref{thm:forking-for-structures}.
  However, we present a separate proof of \Cref{thm:forking-characterization-graphs}, as it gives a more precise description of the employed flips, which can be then used 
  to prove the second main result of this section, giving a closely related, alternative characterization of forking independence over monadically stable graphs 
in terms of a certain auxiliary graph 
(see \Cref{sec:metric-and-disc} for a more precise formulation).
This involves an auxiliary (symmetric) binary relation $\leftrightsquigarrow_S$
on the vertices of $G$, and its transitive, reflexive closure  $\leftrightsquigarrow_S^*$, defined in \Cref{sec:metric-and-disc}.

\begingroup
\renewcommand\thetheorem{\ref{thm:discrepancy}} 
\begin{restatable}{theorem}{discrepancy}
    Let $G \preceq H$ be two monadically stable graphs.
    For two vertices $u, v \in H$ we have $u \nfind G v$ if and only if $u \leftrightsquigarrow_G^* v$ holds.
\end{restatable}
\endgroup



\subsection{Flipping and forking in monadically stable graphs}
Towards a proof of \Cref{thm:forking-characterization-graphs}, we first  establish that  flip independence 
enjoys some finitary analogues of the properties enjoyed by forking independence.

One of the standard properties of forking is called \emph{normality}.
In full generality, it says that for any three sets $A, B, C$ we have
\[
    A \find B C \iff A \find B BC \iff AB \find B C.
\]
From the definition of flip independence it is not clear whether it satisfies a similar property.
For instance, for any graphs $G \preceq H$, vertex $v \in H \setminus G$, and $r \in \N$ we obviously have $v \ind Gr \emptyset$, while showing that $v \ind Gr G$ is a nontrivial objective, and is established by the following fundamental result of \cite{flipper-game}.
\begin{lemma}[{\cite[Lemma 7.1]{flipper-game}}]
    \label{lem:separation-lemma}
    Let $G \preceq H$ be monadically stable graphs and let $v \in H \setminus G$.
    Then, $v \ind Gr G$ for every $r \in \N$.
\end{lemma}

We prove the following,  stronger version of normality of the flip independence relation, but only over models (that is, when $B$ induces an elementary substructure).
The proof is based on refining the proof of \cite[Lemma 7.1]{flipper-game}, and is deferred to Appendix \ref{app:graphs}.

\begin{restatable}{lemma}{graphsnormality}
    \label{lem:flip-independence-normality}
    Let $G \preceq H$ be monadically stable graphs, let $u, v \in H \setminus G$, and let $r \in \N$ be a positive integer.
    If $u \ind G{3r} v$ then $u \ind G{r} vG$.
\end{restatable}

By \Cref{lem:flip-independence-normality} and the symmetry of flip independence we know that whenever $u \ind Gr v$ for every $r \in \N$ then both $u \ind Gr Gv$ and $uG \ind Gr v$ for every $r \in \N$.
A priori it is not clear that there is a single $G$-definable flip $H'$ in which both $\dist(u, Gv)$ and $\dist(v, Gu)$ are large.
However, this is indeed true, as stated below; see Appendix \ref{app:graphs} for a proof.

\begin{restatable}{lemma}{flipdoublesep}
    \label{lem:flip-double-separation}
    Let $G \preceq H$ be monadically stable graphs, let $u, v \in H \setminus G$, and let $r \in \N$ be a positive integer.
    If $u \ind G{3r} Gv$ then there is a $G$-definable flip $H'$ of $H$ such that $\dist_{H'}(u, Gv) > r$ and $\dist_{H'}(v, Gu) > r$.
\end{restatable}

\Cref{thm:forking-characterization-graphs} follows 
from \Cref{lem:separation-lemma,lem:flip-independence-normality,lem:flip-double-separation} and \Cref{cor:fo-locality},
see Appendix~\ref{app:graphs}. A proof of a  similar result, in a more general context will be shown in \Cref{sec:flip-fork}.

\subsection{Discrepancy graph}\label{sec:metric-and-disc}

Let $G\preceq H$ be monadically stable graphs.
 Our aim now is to get a more detailed description of the equivalence relation $\nfind G$.
We will show that it corresponds to the connected components in a certain graph which can be 
combinatorially defined in terms of $G$, as follows.

Recall from \Cref{thm:definability-of-types}
that for $G\preceq H$ and every $u\in H\setminus G$,
there is a (unique, up to equivalence in $H$) quantifier-free formula $\psi_u(y)$ with parameters from $G$ such that 
for all $v\in G$ we have 
$$H\models E(u,v)\iff H\models \psi_u(v).$$
We first define an auxilliary binary relation $E'$ on $V(H)$ 
so that 
$$E'(u,v)\iff \psi_u(v).$$
Intuitively, $E'(u,v)$ means that 
it is ``expected'' that $H\models E(u,v)$, basing 
only on the knowledge of $N(u)\cap G$ and $N(v)\cap G$.
The relation $E'$ is symmetric, as by \Cref{lem:harrington}, we have that $H \models \psi_u(v) \iff H \models \psi_v(u)$ for any two formulas $\psi_u$ and $\psi_v$ that define the atomic types of $u$ and $v$ over $G$ respectively. However, $E'$ may have self-loops, e.g. if $G\preceq H$ are both complete graphs and $v\in H\setminus G$, then $E'(v,v)$ holds.
Note that $E'$ is definable with parameters from $G$.

Next, we define a binary relation $\leftrightsquigarrow_G$ on $V(H)$ as the symmetric difference of $E$ and $E'$, that is, 
    \iflipics
    \[
        u\leftrightsquigarrow_G v  \iff H \models E'(u,v) \xor E(u, v)
        \iff H\models \psi_u(v)\xor E(u,v)
    \]
    \else
    \begin{multline*}
        u\leftrightsquigarrow_G v  \iff H \models E'(u,v) \xor E(u, v)\\
        \iff H\models \psi_u(v)\xor E(u,v)
    \end{multline*}
    \fi
    for $u,v\in V(H)$.
    Observe that the relation $\leftrightsquigarrow_G$ is symmetric,
    since both $E$ and $E'$ are.
    Intuitively, $u\leftrightsquigarrow_G v$ means that there is a discrepancy between the expected and actual relation between $u$ and $v$. 
    Note also that $\neg(u\leftrightsquigarrow_G v)$ holds whenever either $u$ or $v$ is in $G$, by the defining property of $\psi_u(v)$.
    Let $\leftrightsquigarrow_G^*$ denote the transitive, reflexive closure of $\leftrightsquigarrow_G$.

    \discrepancy*



In the special case when $G$ and $H$ are nowhere dense, 
$u\leftrightsquigarrow_G v$ holds if and only if $u$ and $v$ are adjacent in $H$ and 
are not elements of $G$.
Indeed, for every vertex $u \in H \setminus G$, its neighborhood in $G$ is finite \cite[Lemma 1.1]{ivanov}, say equal to $\set{v_1,\ldots,v_k}\subset V(G)$.
Therefore, it can be defined by a formula $\psi_u(y)$ with parameters $v_1,\ldots,v_k$ which just enumerates the set of neighbors of $u$, that is,
$\bigvee_{i=1}^k(y=v_i)$.
In particular, for any $v \in H \setminus G$ we have $H \not \models \psi_u(v)$.
Thus, the relation $\leftrightsquigarrow_G$ in the statement of \Cref{thm:discrepancy} coincides with the edge relation of the graph obtained from $H$  by removing all the edges incident to  $V(G)$.
This agrees with \Cref{thm:ivanov}.

\medskip
To prove \Cref{thm:discrepancy}, we first characterize the relation $\leftrightsquigarrow_G$, in the following lemma, proved in Appendix~\ref{app:graphs}.

\begin{restatable}{lemma}{radiusone}
    \label{lem:radius-1-dependence}
    Let $G \preceq H$ be two monadically stable graphs, let $u, v \in H \setminus G$ be two vertices, and let $\psi_u(x)$ be any formula that defines the neighborhood of $u$ on $G$,
    as given by \Cref{thm:definability-of-types}.
    Then, the following conditions are equivalent:
    \begin{enumerate}
        \item\label{it:rad1-1} $u \nind G1 Gv$;
        \item\label{it:rad1-2} there exists an $S$-flip $H'$ of $H$ for some finite $S \subseteq G$ such that $\dist_{H'}(u, G) > 1$ and $\dist_{H'}(u, v) = 1$;
        \item\label{it:rad1-3} $H \models \psi_u(v) \xor E(u, v)$, equivalently, $u\leftrightsquigarrow_G v$.
    \end{enumerate}
\end{restatable}

We  now prove \Cref{thm:discrepancy}.
\begin{proof}[Proof of \Cref{thm:discrepancy}]
    Assume $u\leftrightsquigarrow_G^*v$. Then $u=u_0 \leftrightsquigarrow_G u_1 \leftrightsquigarrow_G  \cdots \leftrightsquigarrow_G u_{n-1} \leftrightsquigarrow_G u_n=v$ for some $u_0,u_1,\ldots,u_{n-1},u_n\in V(H)$.
    By  \Cref{lem:radius-1-dependence} we have that $u_{i-1} \nind G1 Gu_i$ for $i=1,\ldots,n$. In particular,
    by \Cref{lem:flip-independence-normality},
    $u_{i-1}\nind G3 u_i$ for $i=1,\ldots,n$.
    Finally, by  \Cref{lem:flip-dependence-transitivity} we get that $u \nind G{3n} v$, so $u \nfind G v$ by \Cref{thm:forking-characterization-graphs}.

    For the converse implication, by  \Cref{thm:forking-characterization-graphs} we have that $u \nfind G v$ implies that there is a positive integer $r$ such that $u \nind Gr vG$.
    Consider an $S$-flip $H'$ of $H$ for some finite $S \subseteq G$ such that $\dist_{H'}(u, G) > r$.
    Since $u \nind Gr vG$, we have that $\dist_{H'}(u, v) \le r$.
    It follows that there are vertices $u_1, u_2, \ldots, u_{n-1} \in H$ such that $u, u_1, \ldots, u_{n-1}, v$ is a path in $H'$.
    Moreover, by the triangle inequality, $\dist_{H'}(u_i, G) > 1$ for every $i = 1, \ldots, n - 1$.
    By \Cref{lem:radius-1-dependence} we obtain $u\leftrightsquigarrow_G u_1\leftrightsquigarrow_G u_2 \leftrightsquigarrow_G\cdots\leftrightsquigarrow_G u_{n-1} \leftrightsquigarrow_G v$.
\end{proof}

\section{Nowhere dense relational structures}\label{sec:nd-struc}

The aim of a few next sections is to generalize the results about forking independence in monadically stable graphs to monadically stable structures with relations of arbitrary arity.
In \Cref{sec:nd-struc}, as a motivating example, we consider the case of \emph{nowhere dense relational structures}.
We say that a class $\CC$ of  structures $\CC$ in a finite relational language is \emph{nowhere dense} if the class of Gaifman graphs of structures in $\CC$ is a nowhere dense graph class.

In \Cref{sec:nd-char} we characterize nowhere dense classes of relational structures, generalizing a result of Braunfeld, Dawar, Eleftheriadis, and Papadopoulos \cite{nowhere-dense-structures}.
In \Cref{sec:forking-nd} we lift Ivanov's result, \Cref{thm:ivanov}, from nowhere dense graphs to nowhere dense structures.

\subsection{Nowhere dense classes of relational structures}\label{sec:nd-char}

\sz{this section is just repeating the intro}
In \cite{nowhere-dense-structures} Braunfeld, Dawar, Eleftheriadis, and Papadopoulos proved that for every \emph{monotone} class of structures $\Cc$ in a finite relational language, $\CC$ is monadically dependent if and only if $\CC$ is nowhere dense (cf. \cite[Theorem 1.1]{nowhere-dense-structures}). 
Here, $\CC$ is monotone if for every two structures
$\str A,\str B$ such that 
$V(\str A)\subseteq V(\str B)$ and $R_{\str A}\subseteq R_{\str B}$ for all relation symbols $R$ in the language of $\CC$,
if $\str B\in\CC$ then $\str A\in\CC$.
The \emph{monotone closure} of a class $\CC$ 
is the smallest monotone class containing $\CC$.

In Appendix~\ref{app:nd-struc}, we provide a simple proof of the following result, 
characterizing all nowhere dense classes of relational structures,
generalizing \Cref{thm:nowhere-dense-g} to classes of relational structures.
This immediately implies the main result\footnote{Our result does not imply the hardness result from \cite{nowhere-dense-structures}, stating that model checking is W[1]-hard on every monotone class of relational structures which is not nowhere dense.} of~\cite{nowhere-dense-structures}.

\nwdstr*

\subsection{Forking over models in nowhere dense structures}
\label{sec:forking-nd}
The following is a straightforward generalization the characterization of forking independence over models in nowhere dense graphs implied by \Cref{thm:ivanov}, to the case of structures in a finite relational language with a nowhere dense Gaifman graph.
See Appendix~\ref{app:nd-struc} for a proof.

\fndws*

\section{Flip independence in relational structures}
\label{sec:flip-fork}
The goal of this section is to generalize the results of \Cref{sec:flip-ind-graphs},
relating flip independence and forking independence
in monadically stable graphs, to structures with higher arity relations.
Recall our notion of (definable) flips in 
relational structures:

\defflips*




\medskip

\begin{definition}
    For a relational structure $\str A$ and sets $S,X,Y\subseteq A$, we say that $X$ and $Y$ are \emph{flip independent at radius $r$ over $S$}, denoted
    $X \ind {S}r Y,$ if there is some $S$-flip $\str A'$ of $\str A$ such that there is no path of length at most $r$ in the Gaifman graph of $\str A'$ connecting a vertex in $X \setminus S$ with a vertex in $Y$ or connecting a vertex in $X$ with a vertex in $Y \setminus S$.        
\end{definition}
(We allow $X$ and $Y$ to be single elements, interpreted as singleton sets).
We  prove the following fundamental result, extending \Cref{lem:separation-lemma}.

Recall that every monadically stable class 
is atomic-stable and existentially monadically dependent.
In the end (see \Cref{thm:superloop}), it will follow that the converse implication also holds (this also follows from \cite{braunfeld2022existential}). However, for some arguments to work (specifically, \Cref{lem:closure}, which is later used in \Cref{lem:to-separability}), it turns out that it is convenient to work with the seemingly weaker condition. The following lemma is proved in \Cref{app:flip-fork}.

\begin{restatable}{lemma}{separationhigherarity}
    \label{lem:separation-higher-arity}
    Let $\str M \preceq \str N$ be two existentially monadically dependent and atomic-stable structures in a finite relational language.
    Then $a\ind{M}r  M$ for all $a \in N$ and $r\in\N$.  
\end{restatable}

We will see later in \Cref{sec:apps} 
some consequences of \Cref{lem:separation-higher-arity}
in the finitary setting.
In \Cref{sec:forking-stable-structures}, using \Cref{lem:separation-higher-arity} (and its variant, \Cref{lem:higher-arity-normality}),
we obtain the following, combinatorial characterization of forking independence over models in monadically stable structures in a finite relational language.

\forkingforstructures*
Similarly as in \Cref{lem:flip-dependence-transitivity}, the flip dependence relation is transitive for structures of higher arity, in the following sense.

\begin{lemma}
    \label{lem:structures-flip-dependence-transitivity}
    Let $\str N$ be a structure in a finite relational language, let $S\subseteq N$, and let $u, v, w \in N$, and let $r, q \in \N$ be positive integers.
    If $u \nind {S}{r} v$ and $v \nind {S}{q} w$ then $u \nind {S}{r + q} w$.
\end{lemma}
\begin{proof}
    The same as for \Cref{lem:flip-dependence-transitivity}.
\end{proof}

\Cref{lem:structures-flip-dependence-transitivity} can be seen as the triangle inequality, which shows that
\[
    \fdist_{S}(u, v) \coloneqq \min \setof{r}{u \nind {S}r v}
\]
is a metric on $N$ (with values in $\N\cup\set{\infty}$).
 \Cref{thm:forking-for-structures} establishes that
 if $\str M\preceq \str N$ are monadically stable structures, 
 the non-independence relation $u \nfind {M} v$ for $u,v\in N$ is equivalent to $\fdist_{M}(u, v) < +\infty$. 
In particular, the relation $\nfind {M}$ is an equivalence relation.
This reproves a result of \cite[Section 4]{baldwin-shelah}.

\subsection{Proof of \Cref{thm:forking-for-structures}}
\label{sec:forking-stable-structures}

\Cref{lem:separation-higher-arity} can be strengthened,
in the same way as we obtained \Cref{lem:flip-double-separation},
to prove the following.

\begin{restatable}{lemma}{highnormality}
    \label{lem:higher-arity-normality}
    Let $\str M \preceq \str N$ be two existentially monadically dependent and atomic-stable structures in a finite relational language.
    Let $a, b \in N \setminus M$ be two elements and let $r$ be a nonnegative integer $r$.
    There exists a constant $r'$ that depends only on $\str N$ and $r$ such that whenever $a \ind {M}{r'} b$ then there exists a definable flip $\str N'$ of $\str N$ such that $\dist_{\str N'}(a, bM)>r$ and $\dist_{\str N'}(b, aM) > r$.
\end{restatable}

Flips behaves nicely with respect to being an elementary substructure, as expressed in the following straightforward lemma. See Appendix \ref{app:flip-fork} for a proof.

\begin{restatable}{lemma}{flipsubcommute}\label{lem:flipsubcommute}
    Let $\str M \preceq \str N$ be two structures and let $\str N'$ be an $S$-flip of $\str N$ for a set $S \subseteq M$. Then $\str N'[M] \preceq \str N'$, i.e. the set $M$ induces an elementary substructure of $\str N'$.
\end{restatable}

We now prove \Cref{thm:forking-for-structures}, repeated below.

\forkingforstructures*
\begin{proof}[Proof of \Cref{thm:forking-for-structures}]
    We may assume that $\str M$ is over a finite relational signature, 
    since both conditions in the claimed equivalence reduce to considering every reduct of $\str M$ over a finite relational language.

    Assume that $a \nind {M}r b$ for some  $r\in\N$.
    By \Cref{lem:separation-higher-arity}, we have that $a \ind {M} {2r} M$.
    Therefore, there is a finite set $S\subset G$ and an $S$-flip $\str N'$ of $\str N$ satisfying $\dist_{\str N'}(a,M)>2r$.
    Since $a \nind { M}r b$,
    we have that $\dist_{\str N'}(a,b)\le r$,
    which implies $\dist_{\str N'}(b,M)>r$.
    Since $\str N'$ is an $S$-flip of $\str N$,
    in particular,
    $\str N'$ is definable in $\str N'$, using parameters from $S$.
    Therefore, 
    as the relation $\dist_{{\str N'}}(x,y)\le r$ is definable in $\str N$,
    it follows that there is some first-order formula $\phi(x,y)$ with parameters from $S$ such that 
    $\str N\models \phi(c,d)$
    if and only if $\dist_{\str N'}(c,d)\le r$,
    for all $c,d\in N$.
    Observe that $\str N\models\phi(a,b)$.
    On the other hand, there is no $a'\in M$ such that 
    $\str N\models \phi(a',b)$, as we have $\dist_{\str N'}(b,M)>r$.
    Therefore, $a\nfind M b$.
    \medskip
    
    Now assume that $a \ind { M}r b$ for all $r\in\N$. 
    We prove that $a\find M b$.
    To this end, pick some formula $\phi(x; y)$
    with parameters from $M$ such that 
    $\str N\models \phi(a,b)$. 
    Let $q$ denote the quantifier rank of $\phi$.
    By \Cref{lem:higher-arity-normality}, we have there is  an $M$-flip $\str N'$ of $\str N$ such that
     $$\dist_{\str N'}(a,bM)>7^q,\qquad \dist_{\str N'}(b,M)>7^q.$$
    Since $a\ind{ M}{7^q} bM$, there is
    As $\str N$ is interpretable in $\str N'$ via a quantifier-free formula with parameters from $\str M$,
    we can rewrite $\phi(x, y)$ into a formula $\phi'(x, y)$ of quantifier rank $q$ (with parameters from $M$) such that 
    \[
        \str N \models \phi(c, d) \iff \str N' \models \phi'(c, d)\qquad\text{for all $c,d\in N$.}
    \]
    Let $\tup w\subset M$ denote the parameters occurring in $\phi'$. Note that $\dist_{\str N'}(a,b\tup w)>7^q$, since $\dist_{\str N'}(a,bM)>7^q$.
    By \Cref{cor:fo-locality} applied to the formula $\phi(x; y\tup w)$, there is a formula $\alpha(x)$ such that: 
    \begin{enumerate}
        \item $\str N' \models \alpha(a)$, and 
        \item for all $a'\in N$, if $\str N' \models \alpha(a')$ and $\dist_{\str N'}(a', b\tup w) > q$, then $\str N' \models \phi'(a', b)$.
    \end{enumerate}
    Since 
    $\tp_{\str N'}(a/M)$ is finitely satisfiable in $M$, there is some $a' \in M$ such that $\dist_{\str N'}(a', \tup w) > 7^q$ and $\str N'\models\alpha(a')$.
    Since $\dist_{\str N'}(b, M) > 7^q$, we have $\dist_{\str N'}(a', b)>7^q$
     and so $\dist_{\str N'}(a',b\tup w)>7^q$.
    Therefore, $\str N' \models \phi'(a', b)$, which yields $\str N \models \phi(a', b)$. This proves that $\tp(a / Mb)$ is finitely satisfiable in $M$.
\end{proof}

    We defer the proof of \Cref{prop:forking-for-structures-converse}, which is a converse of \Cref{thm:forking-for-structures} assuming that $\str M$ is stable, to Appendix \ref{app:converse}.

\section{Separation game, flip flatness, and closing the loop of implications}\label{sec:apps}

In this section, we prove the following (see below for definitions). Recall that $\overline \CC$ denotes the elementary closure of the class $\CC$.
\begin{theorem}\label{thm:superloop}
    Let $\CC$ be a class of structures in a finite relational language. The following conditions are equivalent:
    \begin{enumerate}
        \item\label{cond:ms} $\CC$ is monadically stable,
        \item\label{cond:a-md} $\CC$ is atomic-stable and existentially monadically dependent,
        \item\label{cond:sep} every structure $\str M\in\overline\CC$ is $r$-separable for all $r\in\N$,
        \item\label{cond:game} for every $r\in\N$ there exists $k\in\N$ such that Separator wins the Separation Game with radius $r$ in $k$ rounds on every $\str M\in\CC$,
        \item\label{cond:flat} $\CC$ is flip-flat.
    \end{enumerate}
\end{theorem}

We note that the equivalence 
\eqref{cond:ms}$\leftrightarrow$\eqref{cond:a-md} 
(for a slightly different, although equivalent definition 
of monadic stability and monadic dependence, see Appendix \ref{app:converse}) was established in \cite[Theorem 3.21]{braunfeld2022existential}.


\begin{definition}
    We say that a structure $\str M$ in a finite relational language is $r$-separable if for every elementary extension $\str M \preceq \str N$ and an element $a \in N \setminus M$ we have $a \ind{M}{r} M$.
\end{definition}

\begin{lemma}\label{lem:to-separability}
    Let $\CC$ be 
    a monadically dependent and atomic-stable class of structures in a relational language.
    Then every structure $\str M\in \overline \CC$ is $r$-separable, for every $r\in\N$.
\end{lemma}
\begin{proof}
    By \Cref{lem:closure}, $\overline\CC$ 
 monadically dependent and atomic-stable;
 in particular, so is every structure $\str M\in \overline \CC$.
 The conclusion follows by 
 \Cref{lem:separation-higher-arity}.
\end{proof}

Similarly as in \cite{flipper-game} we can define separation game for relational structures.

Fix a radius $r \in \N$. The Separation Game of radius r is played by two players, Separator and Connector, on a relational structure $\str M$, as follows. Let $A_0 = M$ and $S_0 = \emptyset$. For $k = 1, 2, \ldots$, the $k$th
round proceeds as follows.
\begin{itemize}
    \item[--] If $|A_{k - 1}| = 1$, then Separator wins.
    \item[--] Otherwise, Connector picks $c_k \in A_{k-1}$ and we set
    \[
        A_k \coloneqq A_{k-1} - \setof{w}{w \ind{S_{k - 1}}{r} c_k}
    \]
    (where separation is evaluated in the structure $\str M$).
    \item[--] Then Separator picks $s_k \in M$ and we set $S_k \coloneqq S_{k-1} \cup s_k$, and proceed to the next round.
\end{itemize}

We say that \emph{Separator wins the Separation Game} with radius $r$ in $k$ rounds 
on a structure $\str M$, 
to mean that Separator has such a winning strategy in the game.
This is equivalent to the condition $\textup{srk}_r(\str M)<k$ (see \Cref{sec:intro-app}).

In \cite[Theorem 8.1]{flipper-game-arxiv} it is proved that in the case of graphs, separability yields a winning strategy for Separator in the Separation game.
The proof presented there, although states for graphs, is \textit{mutatis mutandis} the same for classes of structures in a finite relational language.
Therefore we get the following theorem.

\begin{lemma}
    \label{lem:to-sep-game}
    Fix $r \in \N$, and let $\CC$ be a class of structures in a finite relational language such that every $\str M \in \overline{\CC}$ is $r$-separable.
    Then there exists $k \in \N$ such that Separator wins the Separation Game with radius $r$ in $k$ rounds on every ${\str M \in \CC}$.
\end{lemma}

From the Separation Game we can go to  flip flatness.
Recall Definition~\ref{def:flip-flat}.
Again, by repeating \textit{mutatis mutandis} \cite[Lemma 4.1]{flipper-game-arxiv} we get the following lemma.

\begin{lemma}
    \label{lem:flip-flatness}
    Let $\CC$ be a class of structure in a finite relational language such that for every $r$ there exists $k$ such that Separator wins the Separation game of radius $r$ on any $\str M \in \CC$, in at
    most $k$ rounds. Then $\CC$ is flip-flat.
\end{lemma}

To close the loop of implications we need to prove that flip-flatness implies monadic stability.
The proof is again repetition \textit{mutatis mutandis} of \cite{flip-flatness}[Lemma 5.3].
Thus we get the following lemma.

\begin{lemma}
    \label{lem:closing-the-loop}
    Every definably flip-flat class of structures in a finite relational language is monadically stable.
\end{lemma}

\begin{proof}[Proof of \Cref{thm:superloop}]
\ 

    \eqref{cond:ms}$\rightarrow$\eqref{cond:a-md} is by \Cref{lem:to-atomic-stable-and-emNIP}.

    \eqref{cond:a-md}$\rightarrow$\eqref{cond:sep} is by \Cref{lem:to-separability}.

    \eqref{cond:sep}$\rightarrow$\eqref{cond:game} is by \Cref{lem:to-sep-game}.

    \eqref{cond:game}$\rightarrow$\eqref{cond:flat} is by \Cref{lem:flip-flatness}.

    \eqref{cond:flat}$\rightarrow$\eqref{cond:ms} is by \Cref{lem:closing-the-loop}.
\end{proof}
\noindent\Cref{thm:flip-flatness-rel} follows.

\begin{szbar}
    \begin{enumerate}
        
        \item zmieścić w 12 stronach
        \item wyrównać bibliografię        
        \item zrobić spellcheck
        \item wrzucić na arxiv
    \end{enumerate}
\end{szbar}

\iflipics\else
\section*{Acknowledgments}
We are grateful to Joanna Fijalkow for many valuable discussions on flips of relational structures.
    \IEEEtriggeratref{21}
\fi
\bibliography{ref}

\begin{thebibliography}{10}

\bibitem{warwick-problems}
{Algorithms, Logic and Structure Workshop in Warwick -- Open Problem Session}.
\newblock
  \url{https://warwick.ac.uk/fac/sci/maths/people/staff/daniel_kral/alglogstr/openproblems.pdf},
  2016.
\newblock [Online; accessed 23-Jan-2023].

\bibitem{AdlerA14}
Hans Adler and Isolde Adler.
\newblock Interpreting nowhere dense graph classes as a classical notion of
  model theory.
\newblock {\em Eur. J. Comb.}, 36:322--330, 2014.
\newblock \href {https://doi.org/10.1016/j.ejc.2013.06.048}
  {\path{doi:10.1016/j.ejc.2013.06.048}}.

\bibitem{baldwin-shelah}
J.T. Baldwin and S.~Shelah.
\newblock Second-order quantifiers and the complexity of theories.
\newblock {\em Notre Dame Journal of Formal Logic}, 26(3):229--303, 1985.

\bibitem{DBLP:conf/lics/BonnetDGKMST22}
{\'{E}}douard Bonnet, Jan Dreier, Jakub Gajarsk{\'{y}}, Stephan Kreutzer,
  Nikolas M{\"{a}}hlmann, Pierre Simon, and Szymon Torunczyk.
\newblock Model checking on interpretations of classes of bounded local
  cliquewidth.
\newblock In Christel Baier and Dana Fisman, editors, {\em {LICS} '22: 37th
  Annual {ACM/IEEE} Symposium on Logic in Computer Science, Haifa, Israel,
  August 2 - 5, 2022}, pages 54:1--54:13. {ACM}, 2022.
\newblock \href {https://doi.org/10.1145/3531130.3533367}
  {\path{doi:10.1145/3531130.3533367}}.

\bibitem{tww4}
\'{E}douard Bonnet, Ugo Giocanti, Patrice~Ossona de~Mendez, Pierre Simon,
  St\'{e}phan Thomass\'{e}, and Szymon Toru\'{n}czyk.
\newblock Twin-width iv: Ordered graphs and matrices.
\newblock {\em J. ACM}, 71(3), June 2024.
\newblock \href {https://doi.org/10.1145/3651151} {\path{doi:10.1145/3651151}}.

\bibitem{tww1}
\'{E}douard Bonnet, Eun~Jung Kim, St\'{e}phan Thomass\'{e}, and R\'{e}mi
  Watrigant.
\newblock Twin-width i: Tractable fo model checking.
\newblock {\em J. ACM}, 69(1), November 2021.
\newblock \href {https://doi.org/10.1145/3486655} {\path{doi:10.1145/3486655}}.

\bibitem{nowhere-dense-structures-icalp}
Samuel Braunfeld, Anuj Dawar, Ioannis Eleftheriadis, and Aris Papadopoulos.
\newblock {Monadic NIP in Monotone Classes of Relational Structures}.
\newblock In Kousha Etessami, Uriel Feige, and Gabriele Puppis, editors, {\em
  50th International Colloquium on Automata, Languages, and Programming (ICALP
  2023)}, volume 261 of {\em Leibniz International Proceedings in Informatics
  (LIPIcs)}, pages 119:1--119:18, Dagstuhl, Germany, 2023. Schloss Dagstuhl --
  Leibniz-Zentrum f{\"u}r Informatik.
\newblock URL:
  \url{https://drops.dagstuhl.de/entities/document/10.4230/LIPIcs.ICALP.2023.119},
  \href {https://doi.org/10.4230/LIPIcs.ICALP.2023.119}
  {\path{doi:10.4230/LIPIcs.ICALP.2023.119}}.

\bibitem{nowhere-dense-structures}
Samuel Braunfeld, Anuj Dawar, Ioannis Eleftheriadis, and Aris Papadopoulos.
\newblock Monadic {NIP} in monotone classes of relational structures.
\newblock In Kousha Etessami, Uriel Feige, and Gabriele Puppis, editors, {\em
  50th International Colloquium on Automata, Languages, and Programming,
  {ICALP} 2023, July 10-14, 2023, Paderborn, Germany}, volume 261 of {\em
  LIPIcs}, pages 119:1--119:18. Schloss Dagstuhl - Leibniz-Zentrum f{\"{u}}r
  Informatik, 2023.
\newblock URL: \url{https://doi.org/10.4230/LIPIcs.ICALP.2023.119}, \href
  {https://doi.org/10.4230/LIPICS.ICALP.2023.119}
  {\path{doi:10.4230/LIPICS.ICALP.2023.119}}.

\bibitem{braunfeld2021characterizations}
Samuel Braunfeld and Michael~C Laskowski.
\newblock Characterizations of monadic {NIP}.
\newblock {\em arXiv preprint arXiv:2209.05120}, 2021.

\bibitem{braunfeld2022existential}
Samuel Braunfeld and Michael~C Laskowski.
\newblock Existential characterizations of monadic {NIP}.
\newblock {\em arXiv preprint arXiv:2209.05120}, 2022.

\bibitem{casanovas}
Enrique Casanovas.
\newblock {\em Simple Theories and Hyperimaginaries}.
\newblock Lecture Notes in Logic. Cambridge University Press, 2011.
\newblock \href {https://doi.org/10.1017/CBO9781139003728}
  {\path{doi:10.1017/CBO9781139003728}}.

\bibitem{Dawar10}
Anuj Dawar.
\newblock Homomorphism preservation on quasi-wide classes.
\newblock {\em J. Comput. Syst. Sci.}, 76(5):324--332, 2010.
\newblock \href {https://doi.org/10.1016/j.jcss.2009.10.005}
  {\path{doi:10.1016/j.jcss.2009.10.005}}.

\bibitem{ms-mc2}
Jan Dreier, Ioannis Eleftheriadis, Nikolas M{\"{a}}hlmann, Rose McCarty,
  Micha{\l} Pilipczuk, and Szymon Toru{\'n}czyk.
\newblock First-order model checking on monadically stable graph classes.
\newblock In {\em 65th {IEEE} Annual Symposium on Foundations of Computer
  Science, {FOCS} 2024, Chicago, IL, USA, October 27-30, 2024}. {IEEE}, 2024.

\bibitem{ms-mc1}
Jan Dreier, Nikolas M{\"{a}}hlmann, and Sebastian Siebertz.
\newblock First-order model checking on structurally sparse graph classes.
\newblock In Barna Saha and Rocco~A. Servedio, editors, {\em Proceedings of the
  55th Annual {ACM} Symposium on Theory of Computing, {STOC} 2023, Orlando, FL,
  USA, June 20-23, 2023}, pages 567--580. {ACM}, 2023.
\newblock \href {https://doi.org/10.1145/3564246.3585186}
  {\path{doi:10.1145/3564246.3585186}}.

\bibitem{flip-flatness}
Jan Dreier, Nikolas M\"{a}hlmann, Sebastian Siebertz, and Szymon Toru\'{n}czyk.
\newblock {Indiscernibles and Flatness in Monadically Stable and Monadically
  NIP Classes}.
\newblock In Kousha Etessami, Uriel Feige, and Gabriele Puppis, editors, {\em
  50th International Colloquium on Automata, Languages, and Programming (ICALP
  2023)}, volume 261 of {\em Leibniz International Proceedings in Informatics
  (LIPIcs)}, pages 125:1--125:18, Dagstuhl, Germany, 2023. Schloss Dagstuhl --
  Leibniz-Zentrum f{\"u}r Informatik.
\newblock URL:
  \url{https://drops.dagstuhl.de/entities/document/10.4230/LIPIcs.ICALP.2023.125},
  \href {https://doi.org/10.4230/LIPIcs.ICALP.2023.125}
  {\path{doi:10.4230/LIPIcs.ICALP.2023.125}}.

\bibitem{flip-breakability}
Jan Dreier, Nikolas M{\"{a}}hlmann, and Szymon Toru{\'n}czyk.
\newblock Flip-breakability: {A} combinatorial dichotomy for monadically
  dependent graph classes.
\newblock In Bojan Mohar, Igor Shinkar, and Ryan O'Donnell, editors, {\em
  Proceedings of the 56th Annual {ACM} Symposium on Theory of Computing, {STOC}
  2024, Vancouver, BC, Canada, June 24-28, 2024}, pages 1550--1560. {ACM},
  2024.
\newblock \href {https://doi.org/10.1145/3618260.3649739}
  {\path{doi:10.1145/3618260.3649739}}.

\bibitem{flipbreakability-arxiv}
Jan Dreier, Nikolas M{\"a}hlmann, and Szymon Toru{\'n}czyk.
\newblock Flip-breakability: A combinatorial dichotomy for monadically
  dependent graph classes.
\newblock {\em CoRR}, abs/2403.15201, 2024.
\newblock \href {http://arxiv.org/abs/2403.15201} {\path{arXiv:2403.15201}}.

\bibitem{dvorakInducedSubdivisions}
Zden{\v e}k Dvo{\v r}{\'a}k.
\newblock Induced subdivisions and bounded expansion.
\newblock {\em Eur. J. Comb.}, 69(C):143--148, mar 2018.
\newblock \href {https://doi.org/10.1016/j.ejc.2017.10.004}
  {\path{doi:10.1016/j.ejc.2017.10.004}}.

\bibitem{gaifman}
Haim Gaifman.
\newblock On local and non-local properties.
\newblock In J.~Stern, editor, {\em Proceedings of the Herbrand Symposium},
  volume 107 of {\em Studies in Logic and the Foundations of Mathematics},
  pages 105--135. Elsevier, 1982.
\newblock URL:
  \url{https://www.sciencedirect.com/science/article/pii/S0049237X08718792},
  \href {https://doi.org/https://doi.org/10.1016/S0049-237X(08)71879-2}
  {\path{doi:https://doi.org/10.1016/S0049-237X(08)71879-2}}.

\bibitem{flipper-game}
Jakub Gajarsk\'{y}, Nikolas M\"{a}hlmann, Rose McCarty, Pierre Ohlmann,
  Micha{\l} Pilipczuk, Wojciech Przybyszewski, Sebastian Siebertz, Marek
  Soko{\l}owski, and Szymon Toru\'{n}czyk.
\newblock {Flipper Games for Monadically Stable Graph Classes}.
\newblock In Kousha Etessami, Uriel Feige, and Gabriele Puppis, editors, {\em
  50th International Colloquium on Automata, Languages, and Programming (ICALP
  2023)}, volume 261 of {\em Leibniz International Proceedings in Informatics
  (LIPIcs)}, pages 128:1--128:16, Dagstuhl, Germany, 2023. Schloss Dagstuhl --
  Leibniz-Zentrum f{\"u}r Informatik.
\newblock URL:
  \url{https://drops.dagstuhl.de/entities/document/10.4230/LIPIcs.ICALP.2023.128},
  \href {https://doi.org/10.4230/LIPIcs.ICALP.2023.128}
  {\path{doi:10.4230/LIPIcs.ICALP.2023.128}}.

\bibitem{flipper-game-arxiv}
Jakub Gajarsk{\'y}, Nikolas M{\"a}hlmann, Rose McCarty, Pierre Ohlmann,
  Micha{\l} Pilipczuk, Wojciech Przybyszewski, Sebastian Siebertz, Marek
  Soko{\l}owski, and Szymon Toru{\'n}czyk.
\newblock Flipper games for monadically stable graph classes, 2023.
\newblock URL: \url{https://arxiv.org/abs/2301.13735}, \href
  {https://doi.org/10.48550/ARXIV.2301.13735}
  {\path{doi:10.48550/ARXIV.2301.13735}}.

\bibitem{tww-tournaments}
Colin Geniet and St\'{e}phan Thomass\'{e}.
\newblock {First Order Logic and Twin-Width in Tournaments}.
\newblock In Inge~Li G{\o}rtz, Martin Farach-Colton, Simon~J. Puglisi, and
  Grzegorz Herman, editors, {\em 31st Annual European Symposium on Algorithms
  (ESA 2023)}, volume 274 of {\em Leibniz International Proceedings in
  Informatics (LIPIcs)}, pages 53:1--53:14, Dagstuhl, Germany, 2023. Schloss
  Dagstuhl -- Leibniz-Zentrum f{\"u}r Informatik.
\newblock URL:
  \url{https://drops.dagstuhl.de/entities/document/10.4230/LIPIcs.ESA.2023.53},
  \href {https://doi.org/10.4230/LIPIcs.ESA.2023.53}
  {\path{doi:10.4230/LIPIcs.ESA.2023.53}}.

\bibitem{GRAVIER2004719}
Sylvain Gravier, Fr{\'e}d{\'e}ric Maffray, J{\'e}r{\^o}me Renault, and Nicolas
  Trotignon.
\newblock Ramsey-type results on singletons, co-singletons and monotone
  sequences in large collections of sets.
\newblock {\em European Journal of Combinatorics}, 25(5):719--734, 2004.
\newblock URL:
  \url{https://www.sciencedirect.com/science/article/pii/S0195669803001586},
  \href {https://doi.org/https://doi.org/10.1016/j.ejc.2003.10.004}
  {\path{doi:https://doi.org/10.1016/j.ejc.2003.10.004}}.

\bibitem{logic-graphs-algorithms}
Martin Grohe.
\newblock Logic, graphs, and algorithms.
\newblock {\em Electron. Colloquium Comput. Complex.}, {TR07-091}, 2007.
\newblock URL:
  \url{https://eccc.weizmann.ac.il/eccc-reports/2007/TR07-091/index.html}.

\bibitem{gks}
Martin Grohe, Stephan Kreutzer, and Sebastian Siebertz.
\newblock Deciding first-order properties of nowhere dense graphs.
\newblock {\em J. {ACM}}, 64(3):17:1--17:32, 2017.
\newblock \href {https://doi.org/10.1145/3051095} {\path{doi:10.1145/3051095}}.

\bibitem{ivanov}
A.~Ivanov.
\newblock The structure of superilat graphs.
\newblock {\em Fundamenta Mathematicae}, 143(2):107–117, 1993.

\bibitem{sparsity-book}
Jaroslav Ne{\v s}et{\v r}il and Patrice~Ossona de~Mendez.
\newblock {\em Sparsity - Graphs, Structures, and Algorithms}, volume~28 of
  {\em Algorithms and combinatorics}.
\newblock Springer, 2012.
\newblock \href {https://doi.org/10.1007/978-3-642-27875-4}
  {\path{doi:10.1007/978-3-642-27875-4}}.

\bibitem{rankwidth-meets-stability}
Jaroslav Ne\v{s}et\v{r}il, Patrice~Ossona de~Mendez, Micha\l{} Pilipczuk, Roman
  Rabinovich, and Sebastian Siebertz.
\newblock Rankwidth meets stability.
\newblock In {\em Proceedings of the Thirty-Second Annual ACM-SIAM Symposium on
  Discrete Algorithms}, SODA '21, pages 2014--2033, USA, 2021. Society for
  Industrial and Applied Mathematics.

\bibitem{NesetrilM11a}
Jaroslav Ne\v{s}et\v{r}il and Patrice {Ossona de Mendez}.
\newblock On nowhere dense graphs.
\newblock {\em Eur. J. Comb.}, 32(4):600--617, 2011.
\newblock \href {https://doi.org/10.1016/j.ejc.2011.01.006}
  {\path{doi:10.1016/j.ejc.2011.01.006}}.

\bibitem{pillay}
A.~Pillay.
\newblock {\em Geometric Stability Theory}.
\newblock Oxford logic guides. Clarendon Press, 1996.

\bibitem{Pillay1986-PILAIT}
Anand Pillay.
\newblock An introduction to stability theory.
\newblock {\em Journal of Symbolic Logic}, 51(2):465--467, 1986.
\newblock \href {https://doi.org/10.2307/2274072} {\path{doi:10.2307/2274072}}.

\bibitem{podewski1978stable}
Klaus-Peter Podewski and Martin Ziegler.
\newblock Stable graphs.
\newblock {\em Fundamenta Mathematicae}, 100(2):101--107, 1978.

\bibitem{Shelah1986}
Saharon Shelah.
\newblock {\em Monadic logic: Hanf Numbers}, pages 203--223.
\newblock Springer Berlin Heidelberg, Berlin, Heidelberg, 1986.
\newblock \href {https://doi.org/10.1007/BFb0098511}
  {\path{doi:10.1007/BFb0098511}}.

\bibitem{tent-ziegler}
Katrin Tent and Martin Ziegler.
\newblock {\em A Course in Model Theory}.
\newblock Lecture Notes in Logic. Cambridge University Press, 2012.
\newblock \href {https://doi.org/10.1017/CBO9781139015417}
  {\path{doi:10.1017/CBO9781139015417}}.

\bibitem{flip-width}
Szymon Toru{\'n}czyk.
\newblock Flip-width: Cops and robber on dense graphs.
\newblock In {\em 64th {IEEE} Annual Symposium on Foundations of Computer
  Science, {FOCS} 2023, Santa Cruz, CA, USA, November 6-9, 2023}, pages
  663--700. {IEEE}, 2023.
\newblock \href {https://doi.org/10.1109/FOCS57990.2023.00045}
  {\path{doi:10.1109/FOCS57990.2023.00045}}.

\end{thebibliography}

\newpage
\iflipics\appendix\else\appendices\fi

\section{Omitted proofs from \Cref{sec:prelims}}
\label{app:prelims}

\closure*
\begin{proof}
    We start with the case of atomic stability.
    Take an atomic formula $\alpha(\tup x; \tup y)$.
    Since it is stable in $\Cc$, there is a bound $K$ such that for any structure $\str A \in \Cc$ there are no $K$ tuples $\tup a_1, \ldots, \tup a_K \in A^{\tup x}$ and $\tup b_1, \ldots, \tup b_K \in A^{\tup y}$ satisfying \[\str A \models \alpha(\tup a_i; \tup b_j) \iff i < j.\]
    Observe that this is a property expressible in FO.
    Since it is satisfied by every structure in $\Cc$, it is also satisfied by every structure in $\overline \Cc$.
    As $\alpha$ was arbitrary, we get that $\overline \Cc$ is atomic-stable.

    We prove the case of existential monadic dependence by contrapositive.
    Assume that $\overline \Cc$ is not existentially monadically dependent.
    Therefore there exists an existential formula $\phi(\tup x; \tup y) \equiv \exists_{\tup z}. \beta(\tup x, \tup y, \tup z)$ where $\beta$ is a quantifier-free formula, such that for every $K$ there exists a structure $\str A \in \overline \Cc$, a monadic lift $\str A^+$ of $\str A$ and two sequences of tuples $(\tup a_i: i \in [K])$ and $(\tup b_J: J \subseteq [K])$ such that \[\str A^+ \models \phi(\tup a_i; \tup b_J) \iff i \in J.\]
    Consider a subset $S$ of the universe $A$ that consists of:
    \begin{itemize}
        \item all the elements appearing in tuples $\tup a_i$ and $\tup b_j$; and
        \item for every $i \in J$ elements of a tuple $\tup c$ such that $\str A^+ \models \beta(\tup a, \tup b, \tup c)$.
    \end{itemize}
    Clearly, $|S|$ can be bounded in terms of $K$ and $\phi$.
    Since $\phi$ is an existential formula we get that \[\str A^+[S] \models \phi(\tup a_i; \tup b_j) \iff i \in J.\]

    Observe that there is a structure in $\Cc$ that has an induced substructure isomorphic to $\str A[S]$.
    Indeed, since $S$ is finite, there is an FO sentence $\psi$ expressing that a structure has an induced substructure isomorphic to $\str A[S]$.
    If $\psi$ was false in every structure in $\Cc$, then it would be false also in every structure in $\overline \Cc$.
    Thus, we get a structure $\str B \in \Cc$ and a subset of its universe $S' \subseteq B$ such that $\str B[S']$ is isomorphic to $\str A[S]$.

    Take a unary predicate $U$ that doesn't appear in $\phi$ and consider a slightly modified formula
    \[
        \widehat{\phi}(\tup x; \tup y) \coloneqq \exists_{\tup z}. \beta(\tup x, \tup y, \tup z) \land \bigwedge_{z \in \tup z} U(z).
    \]
    Observe that we can color $\str B$ to obtain $\str B^+$ so that $\str B^+[S']$ is isomorphic to $\str A^+[S]$.
    Additionally, interpret predicate $U$ in $\str B^+$ as the set $S'$, thus obtaining a structure $\str B^{++}$.
    Clearly, we can find in $B$ tuples $(\tup a_i' : i \in [K])$ and $(\tup b_J' : J' \subseteq [K])$ (those that are mapped by the isomorphism to $\tup a_i$ and $\tup b_J$ respectively) such that
    \[
        \str B^{++} \models \widehat\phi(\tup a_i'; \tup b_J') \iff i' \in J'.
    \]
    Since $K$ was arbitrary we get that $\widehat\phi$ is not monadically dependent in $\Cc$, which is a contradiction.
\end{proof}

\dvorak*

\begin{proof}
    Suppose $\CC$ is not nowhere dense.
    Then there is a number $\ell$ such that 
    for every $t\in \N$ there is a graph $G \in \CC$ such that an $\ell$-subdivision of $K_{t, t}$ is a subgraph of $G$.
    Since $\CC$ is weakly sparse, we get that $\ell > 0$.
    Let us denote the $t$ vertices of $G$ corresponding to one part of the biclique $K_{t,t}$ by $v_1, \ldots, v_t$, and the $t$ vertices corresponding to the other part of $K_{t,t}$ by $u_1, \ldots, u_t$.
    Let us also denote vertices on the path between $v_i$ and $u_j$ by $w_{ij}^1, \ldots, w_{ij}^\ell$.

    By applying \Cref{thm:bi-ramsey-tuples} (and potentially going to a subgraph) we can assume that whether or not there is an edge between $w_{ij}^k$ and $w_{i'j'}^{k'}$ depends only on $k, k', \otp(i, i')$, and $\otp(j, j')$.
    If for some $k, k'$ (potentially equal) there are order types (other than $i = i'$ and $j = j'$) for which we have edges between $w_{ij}^k$ and $w_{i'j'}^{k'}$ we get a contradiction with $\CC$ being weakly sparse.
    The only case in which we don't get a contradiction is $i = i'$ and $j = j'$. However, if $k$ and $k'$ are not two consecutive integers, than we can find an $\ell'$-subdivision of $K_{t, t}$ for $\ell' < \ell$.

    Again by applying \Cref{thm:bi-ramsey-tuples} (and potentially going to a subgraph) we can assume that whether or not there is an edge between $v_i$'s and $w_{i'j'}^k$ depends only on $k$ and $\otp(i, i')$.
    Similarly as before, if we have edges between $v_i$ and $w_{i'j'}^k$ for some $k$ and every $i' > i$ we would get a large biclique -- a contradiction with $\CC$ being weakly sparse.
    On the other hand, the case $i = i'$ for $k > 1$ means we can take smaller $\ell$.
    Therefore, we get that we can find an $\ell$-subdivision of $K_{t, t}$ for some large $t$ and $\ell > 0$ as an induced subgraph of $G$.

    Now observe that we can define the $\ell$-subdivision of any bipartite graph in $\Cc$.
    Indeed, we can mark all the principal vertices of our subdivided $K_{t, t}$ with one unary predicate, and with second unary predicate we can mark subdividing vertices on these edges that we wish to keep.
    Since this subdivided $K_{t, t}$ is induced in the Gaifman graph, then we can recover all the selected subdividing paths. 
\end{proof}
\section{Omitted proofs from \Cref{sec:flip-ind-graphs}}
\label{app:graphs}

\subsection{Proof of \Cref{lem:flip-independence-normality}}
In this section, we prove \Cref{lem:flip-independence-normality},
 repeated below.
\graphsnormality*


For the proof of \Cref{lem:flip-independence-normality}, we use Lemma 7.10 from \cite{flipper-game} which we restate below.

\begin{lemma}
    \label{lem:types-of-balls}
    Fix $r \in \N$. Let $G$ be a monadically stable graph, let $H$ be its elementary extension, and let $v \in H$ be such that the $r$-ball $B_r(v)$ around $v$ in $H$ is disjoint from $G$.
    Then $\Types E (B_r(v) / G)$ is finite.
\end{lemma}

We also prove a stronger version of Lemma 7.3 from \cite{flipper-game}.

\begin{lemma}
    \label{lem:normality-sublemma}
    Let $G \preceq H$ be two graphs with stable edge relation and let $U \subseteq H \setminus G$ be a set of vertices such that $\Types E (U/G)$ is finite.
    Then there exists a finite set $S \subseteq G$ and an $S$-definable flip $H'$ of $H$ such that:
    \begin{enumerate}
        \item $\dist_{H'}(U,G)>1$,        
        \item unless there is $u \in U$ which is complete to $G$, then for all $v, w \in H$, if $\dist_{H}(w, S) > 1$ 
        then $H\models E(v,w)\iff H'\models E(v,w)$.
    \end{enumerate}
    Moreover, for every   
     $v \in H \setminus G$,
     \begin{enumerate}[resume]
         \item if $\dist_H(v, U) > r\ge 3$ then $\dist_{H'}(v, U) > r - 2$, and 
         \item if $\dist_H(v, G) > r\ge 1$ then $\dist_{H'}(v, G) > r$.
    \end{enumerate}
\end{lemma}

\begin{proof}
    Let $\set{u_1, \ldots, u_k} \subseteq U$ contain one vertex from each type in $\Types E (U/G)$.
    Let $W_0\subseteq G$ be obtained by taking for each pair $u_i, u_j \in U$ a vertex $w \in G$ that is adjacent to exactly one of $u_i, u_j$.
    Moreover, for every $u_i$ that has a neighbor we take at least one such neighbor to $W_0$.
    
    By \Cref{thm:definability-of-types}, the neighborhood of every $u_i$ in $G$ is definable by a formula $\psi_{u_i}(y)$ which is a positive boolean combination 
    of formulas $E(y,w^i_1),\ldots, E(y,w^i_\ell)$, for some  
    some $w^i_1, \ldots, w^i_\ell \in G$, where $\ell$ depends only on $H$.
    Moreover, we can assume that $\dist_{H}(u_i, w^i_j) \le 2$ for each $1 \le i \le k, 1 \le j \le \ell$ (c.f. \cite[Lemma 2.10]{casanovas}).


    Define $$S \coloneqq W_0\cup \setof{w^i_j}{i\in\set{1,...,k},j\in\set{1,\ldots,\ell}}.$$
    Note that $\dist_H(s,U)\le 2$ for all $s\in S$.
    By an \emph{$S$-class} we mean an equivalence class 
of the partition of $H$ into atomic types over~$S$.
Let $A_i$ denote the $S$-class of $u_i$, for $i=1,\ldots,k$.

We say that a set $X \subseteq V(\str H)$ of vertices is complete (resp. anti-complete) towards a disjoint set of vertices $Y \subseteq V(\str H)$ if for every $x \in X, y \in Y$ we have $\str H \models E(x, y)$ (resp. $\str H \not \models E(x, y)$).

 \begin{claim}\label{cl:homo}
    Fix two $S$-classes $A,B$ such that $A$ intersects $U$.
    Then $U\cap A$ is either complete, or anti-complete towards 
$B\cap G$ in $H$.
 \end{claim}
\begin{claimproof}
Since $A$ is an $S$-class and $W_0\subseteq S$,
it follows that all vertices 
in $U\cap A$ have identical neighborhoods in $G$.
In particular, 
$\set{u_1, \ldots, u_k} \cap A=\set{u_i}$ for some $i\in[k]$.
Furthermore, it is enough to show that $u_i$ is adjacent to either all, or no vertices 
in $B\cap G$, in the graph $H$.

Suppose $b,b'\in B\cap G$ and $u_i$ is adjacent to $b$.
Therefore,  $H\models\psi_{u_i}(b)$.
Since $\psi_{u_i}(y)$ is a quantifier-free formula with parameters from $S$
and $\atp(b/S)=\atp(b'/S)$, it follows that $H\models\psi_{u_i}(b')$.
Therefore, $u_i$ is adjacent to $b'$ in~$H$.
\end{claimproof}

 \begin{claim}\label{cl:symmetry}
    Let $A,B$ be two $S$-classes
    intersecting $U$.
Then $U\cap A$ is complete towards $B\cap G$ 
if and only if $U\cap B$ is complete towards $A\cap G$.
 \end{claim}
 \begin{claimproof}
    Fix $i,j\in[k]$
    such that $u_i\in U\cap A$ and $u_j\in U\cap B$.
    We show that  $u_i$ is complete towards $A_j\cap G$ implies that 
  $u_j$ is complete towards $A_i\cap G$.

    Since $G \preceq H$, there is some $b \in G$ with $\atp(b/S)=\atp(u_j/S)$.
    Therefore, $b \in A_j \cap G$, so it is a neighbor of $u_i$ in $H$.
    Then $H\models\psi_{u_i}(b)$ by \Cref{thm:definability-of-types}.
    As $\atp(b/S)=\atp(u_j/S)$ and $\psi_{u_i}(y)$ is a quantifier-free formula
    with parameters from $S$,
    it follows that $H\models\psi_{u_i}(u_j)$. 
    By \Cref{lem:harrington} we get that 
    $H\models\psi_{u_j}(u_i)$.
    For any $a\in G \cap A_i$ we have 
    $\atp(a/S)=\atp(u_i/S)$.
    Since $\psi_{u_j}(y)$ is a quantifier-free formula with parameters from $S$,
    it follows that $H\models\psi_{u_j}(a)$. This implies that $a\in G\cap A_i$ is a neighbor of $u_j$ in $H$. This proves the claim.
\end{claimproof}

Let $H'$ be the $S$-definable flip of $H$ such that for every 
    two $S$-classes $A,B$,
we have that the pair $A,B$ is flipped, that is,
    $$H\models E(a,b)\iff H'\models \neg E(a,b)$$
holds for all distinct $a\in A,b\in B$,
if and only if
$U\cap A$ is nonempty and is complete towards $G\cap B$, or $U\cap B$ is nonempty and is complete towards $G\cap A$.
We verify that $H'$ satisfies the properties stated in \Cref{lem:normality-sublemma}.
\medskip

Let $T$ denote the unique $S$-class with no neighbors in $S$. 
We argue that unless there is $u \in U$ which is complete to $G$, the $S$-class $T$ is not flipped with any $S$-class,
which proves the second property.
As $T\cap U$ has no neighbors in $W_0$,
it follows that $T\cap U$ has no neighbors in $G$ either. 
In particular, $T\cap U$ is not complete towards $G\cap B$, for all $S$-classes $B$.
Therefore, if $T$ was flipped, then there is an $S$-class $B$ such that $B \cap U$ is nonempty and is complete to $T \cap G$.
Take a vertex $u_i \in B \cap U$ and a vertex $w \in T \cap G$.
Since $H \models E(u_i, w)$ then $H \models \psi_{u_i}(w)$.
As $\atp(w / S) = \emptyset$ and $\psi_{u_i}$ is a positive quantifier-free formula, then $\psi_{u_i}$ is equivalent in $H$ to $\top$.
This implies $u_i$ is complete to $G$.

We verify the first property. Let $u\in U$. 
If $u$ has no neighbors in $G$ in the graph $H$,
then in particular, $u$ has no neighbors in $S$,
and therefore, $u\in T$.
Therefore, $u$ also has no neighbors in $G$ in the graph $H'$, by the second property proved above.
So suppose that $u$ has some neighbor in $G$.
Then there is some $i\in[k]$ such that $\atp(u/G)=\atp(u_i/S)$.
In particular, $u$ is in the $S$-class $A_i$ of $u_i$, and 
 $u$ and $u_i$ have equal neighborhoods in $G$ in the graph $H$.
As $H'$ is an $S$-definable flip of $H$,
the neighborhoods of $u$ and of $u_i$ in $G$ in the graph $H'$ are also equal.
Therefore, it is enough to argue that $u_i$ has no neighbors in $G$ in the graph $H'$.
This is because by \Cref{cl:homo}, the class $A_i$ is complete or anti-complete towards 
$B\cap G$, for every $S$-class $B$,
and $A_i$ is flipped with $B$ if and only if $A_i$ is complete towards $B\cap G$,
by \Cref{cl:symmetry}.

We now verify the third and fourth property.
For every $s \in S$ we have $\dist_{H}(s, U) \le 2$.
Together with $\dist_H(v,U)>r$, this implies  $\dist_{H}(s, v) > r - 2$, and therefore $B^{H}_{r - 3}(v) \subseteq T$.
Moreover, if there is a vertex $u \in U$ that is complete to $G$, then $s$ has to be anticomplete to $G$, since otherwise $\dist_{H}(u, v) = 2$.
This contradicts \Cref{lem:complete-anticomplete-tuples},
which implies that there cannot simultaneously be a vertex that is complete to $G$ and a vertex that is anti-complete to $G$.
Therefore, by the second property, $B^{H}_{r - 3}(v)$ has the same neighborhood both in $H$ and $H'$ (which is disjoint from $U$) and finally we get $\dist_{H'}(v, U) > r - 2$.
Similarly, if $\dist_{H}(v, G) > r\ge 1$, then in particular $B^H_{r-1}(v)\subseteq T$ and there is no $u \in U$ that is complete to $G$, so $B^H_{r-1}(v)$ has the same neighborhood both in $H$ and in $H'$ (which is disjoint from $G$), and we conclude that $\dist_{H'}(v,G)>r$.
\end{proof}

We now prove \Cref{lem:flip-independence-normality}.
\graphsnormality*
\begin{proof}[Proof of \Cref{lem:flip-independence-normality}]
We construct a sequence of finite sets $S_0\subseteq S_1\subseteq \ldots\subset G$,
sets $\set u= U_0\subseteq U_1\subseteq \ldots\subseteq G$ 
and graphs $G^{(0)}\preceq H^{(0)},G^{(1)}\preceq H^{(1)},\ldots$,
such that for $i=0,1,2,\ldots$:
\begin{itemize}
    \item $H^{(i)}$ is an $S_i$-definable flip of $H$,
    \item $G^{(i)}\preceq H^{(i)}$ is the subgraph of $H^{(i)}$ induced by $V(G)$,
    \item  $\dist_{H^{(i)}}(u,v)>3r-2i$ 
    \item the set $\Types{E}(U^{(i)}/G^{(i)})$ (evaluated in $H^{(i)}$) has finite size.
\end{itemize}
As $u \ind G{3r} v$, there is 
     an $S_0$-definable flip $H^{(0)}$ of $H$, for some finite $S_0\subset G$, such that $\dist_{H^{(0)}}(u, v) > 3r$.
    Let $G^{(0)}$ be the subgraph of $H^{(0)}$ induced by $V(G)$; then $G^{(0)}$ is an $S_0$-definable flip of $G$ and $G^{(0)}\preceq H^{(0)}$.

    We proceed by induction on $i=0,1,2,\ldots$.
Apply \Cref{lem:normality-sublemma} to $G^{(i)}\preceq H^{(i)}$ and $U^{(i)}$, obtaining an $S_{i+1}$-definable flip $H^{(i+1)}$ of $H^{(i)}$, for some finite $S_{i+1}\subset G$ (without loss of generality, $S_i\subseteq S_{i+1}$)
    such that $\dist_{H^{(i+1)}}(v, U^{(i)}) > 3r - 2i-2$ and $U_i$ is disjoint from $G$.
In particular, $\dist_{H^{(i+1)}}(v,u)>3r-2(i+1)$.

Let $G_{(i+1)}$ be the subgraph of $H^{(i+1)}$ induced by $V(G)$, and 
    let $$U^{(i+1)} \coloneqq B_1^{H^{(i+1)}}(U_i).$$
    By \Cref{lem:types-of-balls} we have that $\Types{E}(U^{(i+1)}/G^{(i+1)})$ (evaluated in $H^{i+1})$ has finite size.
    This completes the inductive step.

    For $i=r$ we finally obtain a flip $H^{(r)}$ of $H$ in which $\dist_{H^{(r)}}(u, G) > r$ and $\dist_{H^{(r)}}(u, v) > 3r - 2r = r$.
\end{proof}

\subsection{Proof of \Cref{thm:forking-characterization-graphs}}
Before proving \Cref{thm:forking-characterization-graphs}, we prove the following.
\flipdoublesep*
\begin{proof}
    Since $u \ind G{3r} Gv$ there is a flip $H''$ of $H$ such that $\dist_{H''}(u, Gv) > 3r$.
    By iteratively applying \Cref{lem:normality-sublemma} we obtain a flip $H'$ of $H''$ such that $v$ is at distance more than $r$ from $G$ and at the same time the distance from $u$ to $Gv$ stays above $r$.
\end{proof}

We are ready to prove \Cref{thm:forking-characterization-graphs},
repeated below.
\forkinggraphs*
\begin{proof}
    First, assume that $u \nind Gr v$ for some $r$.
    By \Cref{lem:separation-lemma} we have that $u \ind G{2r} G$. Therefore, there is
     a finite set $S \subset G$ and
    an $S$-definable flip $H'$ of $H$ satisfying $\dist_{H'}(u, G) > 2r$.
    As $u \nind Gr v$, in particular $\dist_{H'}(u, v) \le r$, which implies $\dist_{H'}(v, G) > r$.
    Since the flip $H'$ is definable using $S$, 
    there is a first-order formula $\phi(x, y)$ with parameters from $S$ that stipulates $\dist_{H'}(x, y) \le r$.
    Now, observe that $H \models \phi(u, v)$, so $\phi(x, v) \in \tp(u/Gv)$.
    On the other hand, there is no $w \in G$ such that $H \models \phi(w, v)$, as we have $\dist_{H'}(v, G) > r$.
    It follows that $\tp(u/Gv)$ is not finitely satisfiable in $G$, so $u \nfind G v$.

    Conversely, assume that $u \ind Gr v$ for every $r$.
    By \Cref{lem:flip-independence-normality} we have that $u \ind Gr vG$ for every $r \in \N$.
We prove that $u\find Gv$. To this end, 
pick any formula $\phi(x, y)$ with parameters $\tup w$ from $G$ such that $H \models \phi(u, v)$;
we need to show that there is some $u'\in G$ 
with $H\models \phi(u',v)$.
    Let $q$ denote the quantifier rank of $\phi$.
    Since $u \ind G{3\cdot7^q} Gv$, by \Cref{lem:flip-double-separation} there is a finite set $S \subset G$ and an $S$-definable flip $H'$ of $H$ such that both $\dist_{H'}(u, Gv) >  7^q$ and $\dist_{H'}(v, Gu) > 7^q$.
    Consider a unary lift $\hat H'$ of $H'$ which colors every element of $H$ based on its atomic type over $S$ in $H$. Define a unary lift $\hat G'$ of $G$ analogously.
    Since the coloring is definable with parameters from $G$, we have that $\hat G' \preceq \hat H'$.
    Since the edge relation of $H$ can be defined in $\hat H'$ by a quantifier-free formula, there is a formula $\phi'(x, y)$ of the same quantifier rank as $\phi(x,y)$, namely $q$, with parameters $\tup w$ and in the language of $\hat H'$, such that
    \begin{align}\label{eq:flip-rewriting}
        H \models \phi(x, y) \iff \hat H' \models \phi'(x, y).    
    \end{align}
    By \Cref{cor:fo-locality}, there is a first-order formula $\alpha(x)$ without parameters such that for any $u'$ that satisfies $\dist_{H'}(u', v\tup w) > 7^q$ we have $\hat H' \models \alpha(u')$ if and only if $\hat H' \models \phi'(u', v)$. 
As $\dist_{H'}(u,Gv)>7^q$ and $\hat H'\models \phi'(u,v)$ by \eqref{eq:flip-rewriting}, in particular $\hat H'\models \alpha(u)$.
    Since $\tp_{\hat H'}(u/G)$ is finitely satisfiable in $G$ there is a vertex $u' \in G$ such that $\hat H' \models \alpha(u')$ and $\dist_{H'}(u', \tup w) > 7^q$.
    Moreover, as $\dist_{H'}(v, G) > 7^q$, in particular $\dist_{H'}(u', v) > 7^q$.
    It follows that $\hat H' \models \phi'(u', v)$, so $H \models \phi(u', v)$ by \eqref{eq:flip-rewriting}.
    This proves $\tp(u/Gv)$ is finitely satisfiable in $G$, so $u \find G v$.
\end{proof}

\subsection{Proof of \Cref{lem:radius-1-dependence}}
\radiusone*
\begin{proof}
    \eqref{it:rad1-1}$\rightarrow$\eqref{it:rad1-2}.
    Let $H'$ be an $S$-definable flip of $H$ for some $S \subseteq G$ such that $\dist_{H'}(u, G) > 1$.
    Since $u \nind G1 Gv$ we have $\dist_{H'}(u, v) = 1$.

    \eqref{it:rad1-2}$\rightarrow$\eqref{it:rad1-3}.
    Condition $\dist_{H'}(u, G) > 1$ means that we flipped in $H$ the $S$-class which contains $u$ (denote it by $U$) with every $S$-class that contains a neighbor of $u$ in $G$.
    In particular, for every $S$-class $T$ we have that $u$ is either complete or anti-complete to $T \cap G$.
    Observe that $G \preceq H$ implies that for every nonempty $S$-class $T$ the set $T \cap G$ is also nonempty.
    Therefore, we flipped $U$ with every $S$-class $T$ such that $u$ is complete to $T \cap G$.
    Denote the $S$-class of $v$ by $V$.
    We obtain that $\dist_{H'}(u, v) = 1$ if either $E(u, v)$ and $u$ is anti-complete to $V \cap G$ or $\neg E(u, v)$ and $u$ is complete to $V \cap G$.
    If $u$ is complete to $T \cap G$ then $H \models \psi_u(w)$ for every $w \in T \cap G$.
    As $S$-classes are definable with parameters from $G$ and $G \preceq H$, then $u$ being complete to $T \cap G$ implies $H \models \psi_u(w)$ for every $w \in T$.
    Similarly, $u$ being anti-complete to $T \cap G$ implies $H \not\models \psi_u(w)$ for every $w \in T$.
    
    \eqref{it:rad1-3}$\rightarrow$\eqref{it:rad1-1}.
    Assume by contradiction $u \ind G1 Gv$.
    Therefore, there exists an $S$-definable flip $H'$ of $H$ for some finite $S \subseteq G$ such that $\dist_{H'}(u, G) > 1$ and $\dist_{H'}(u, v) > 1$.
    Similarly as in the previous implication, we observe that in order to obtain $H'$ we flipped the $S$-class of $u$ (denoted $U$) with every $S$-class $T$ such that $H \models \psi_u(w)$ for every $w \in T$.
    Let $V$ denote the $S$-class of $v$.
    Since $\dist_{H'}(u, v) > 1$ either $E(u, v)$ and $H \models \psi_u(v)$ or $\neg E(u, v)$ and $H \not \models \psi_u(v)$, a contradiction.
\end{proof}
\section{Omitted proofs from \Cref{sec:nd-struc}}
\label{app:nd-struc}

\subsection{Proof of \Cref{thm:nowhere-dense-struct}}

\nwdstr*
The following lemma yields the implication \eqref{it:nds1}$\rightarrow$\eqref{it:nds2} in \Cref{thm:nowhere-dense-struct}.
\begin{lemma}[{Follows from \cite[Proposition 5.7]{sparsity-book} and \cite{podewski1978stable}, see \cite[Theorem 30]{nowhere-dense-structures-icalp}}]\label{lem:nd is stable}
    Let $\CC$ be a class of structures in a finite relational language such that 
     $\gaif(\CC)$ is nowhere dense.
    Then $\CC$ is monadically stable.
\end{lemma}
\begin{proof}
    The \emph{incidence graph} $I(\str A)$ of a relational structure $\str A$ 
    is the bipartite graph where one part is $V(\str A)$, and the other part is $\coprod_{R}R_{\str A}$ -- the disjoint union of all relations in the language of $\str A$ -- and a tuple $\tup a\in R_{\str A}$ is adjacent in $I(\str A)$ to all elements in $V(\str A)$ which occur in $\tup a$.
    Let $I'(\str A)$ denote the binary relational structure obtained from $I(\str A)$ by labelling the vertices and edges, so that each 
tuple $(a_1,\ldots,a_k)\in R_{\str A}$ is labelled by the unary predicate $U_R$,
and each edge connecting $\tup a$ with an element $a\in V(\str A)$ 
is labelled by the set $\setof{i\in \set{1,\ldots,k}}{a_i=a}$.
    \cite[Proposition 5.7]{sparsity-book},
     establishes that if $\cal G(\CC)$ is nowhere dense then so is $\cal I(\CC)\coloneqq\setof{I(\str A)}{\str A\in\CC}$.
     Since the class $\cal I'(\CC)\coloneqq\setof{I'(\str A)}{\str A\in\CC}$
     is obtained from the nowhere dense class $I(\CC)$ by labelling the vertices and edges, 
     it follows from the result of \cite[Corollary 10 and Remark]{podewski1978stable}
     that  $\cal I'(\CC)$ is monadically stable.
    It is easy to see that class $\cal I'(\CC)$ transduces $\CC$, 
    see \cite[Section 6]{nowhere-dense-structures-icalp}.
    Therefore, $\CC$ is monadically stable.
\end{proof}

The following lemma yields the implication \eqref{it:nds2}$\rightarrow$\eqref{it:nds3} in \Cref{thm:nowhere-dense-struct}.
\begin{lemma}
    \label{lem:nip-monotone-to-weakly sparse}
    Let $\Cc$ be a monotone and monadically dependent class of structures in a finite relational language.
    Then the class of Gaifman graphs of structures from $\Cc$ is weakly sparse.
\end{lemma}
\begin{proof}
    Assume that for some large $t\in \N$ there is a structure $\str M \in \Cc$ such that the biclique $K_{t, t}$ is an induced subgraph of the Gaifman graph of $\str M$.
    Let us denote elements on one side of this $K_{t, t}$ by $a_1, \ldots, a_t$ and on the other by $c_1, \ldots, c_t$.
    Every edge of this $K_{t, t}$ is witnessed by a tuple which is in some relation in the language.
    By applying \Cref{thm:bi-ramsey-tuples} (for $\ell=1$), and going to a subgraph of $K_{t, t}$ we can assume there is one relation $R$ such that for every $a_i, c_j$ we have a tuple $\tup b_{ij}$ such that $R(a_i, \tup b_{ij}, c_j)$ holds (possibly after permuting arguments of $R$).
    Since our class of structures is monotone we can assume that $a_i \tup b_{ij} c_j$ are the only tuples on which relation $R$ holds and all the other relation don't hold on any tuple.
    Again, by using \Cref{thm:bi-ramsey-tuples} (for $\ell=2$) and going to a subgraph of $K_{t, t}$ we get that the atomic type of the tuple $(a_i, \tup b_{ij}, c_j, a_{i'}, \tup b_{i'j'}, c_{j'})$ for any $i, i', j, j'$ depends only on  $\otp(i, i')$ and  $\otp(j, j')$.

    Clearly, we get that $a_{i'}$ can't be an element of $\tup b_{ij}$ for $i \neq i'$.
    Indeed, if $a_{i'} = \tup b_{ij}^{(k)}$ for $i < i'$ then $a_{i'} = \tup b_{1j}^{(k)}$ and for any $\ell > 1$ we have $a_{\ell} = \tup b_{1j}^{(k)}$ implying that nearly all $a$'s are equal.
    The case $i' < i$ is analogous.
    Similarly, $c_{j'}$ can't be an element of $\tup b_{ij}$ for $j \neq j'$.

    As a result, we can observe that for any $i, j$ the only tuple in the relation $R$ that contains both $a_i$ and $c_j$ is $a_i \tup b_{ij} c_j$.
    Indeed, any other tuple $a_{i'} \tup b_{i'j'} c_{j'}$ for $i \neq i'$ doesn't contain $a_i$ and similarly for $c_j$.
    Therefore, by removing certain tuples from the relation $R$ we can get an arbitrary subgraph of $K_{t, t}$ thus contradicting the class of Gaifman graphs of structures from $\Cc$ being monadically dependent.
\end{proof}

\begin{proof}[Proof of \Cref{thm:nowhere-dense-struct}]
    \eqref{it:nds1}$\rightarrow$\eqref{it:nds2}.
    As $\CC$ is nowhere dense, so is the monotone closure $\CC'$ of $\CC$. Therefore, $\CC'$ is monadically stable by 
    \Cref{lem:nd is stable}.

    \eqref{it:nds2}$\rightarrow$\eqref{it:nds3}. It is enough to prove that if $\CC$ is monotone and monadically stable, then $\GG\coloneqq\gaif(\CC)$ is weakly sparse. This is proved in \Cref{lem:nip-monotone-to-weakly sparse}.
    
The implication \eqref{it:nds3}$\rightarrow$\eqref{it:nds4} is immediate.
    Finally, for the implication \eqref{it:nds4}$\rightarrow$\eqref{it:nds1}, note that if $\CC$ is monadically dependent then so is $\GG$,
    since $\GG$ interprets in $\CC$ via a simple interpretation.
    As $\GG$ is weakly sparse and monadically dependent,
    it follows that $\GG$ is nowhere dense, by \Cref{lem:nip-weakly sparse-to-nowhere-dense}.
    \end{proof}




\subsection{Proof of \Cref{thm:forking-nowhere-dense-structures}}

\fndws*
\begin{proof}
    First, assume that $u, v$ are at distance $r < \infty$ in $\gaif(\str N) \setminus \gaif(\str M)$.
    Since $\gaif(N)$ is nowhere dense, by \cite[Lemma 1.1]{ivanov} there is a finite set of vertices $S \subset M$ which are reachable from $v$ by a path of length at most $r$ with all the internal vertices outside of $M$.
    Consider a formula $\phi(x, y)$ with parameters in $S$ which stipulates ``$x$ and $y$ are connected by a path of length at most $r$ which avoids all vertices from $S$''.
    Clearly $\phi(u, v) \in \tp(u / Mv)$ but for every element $w \in M$ we have $\str N \not \models \phi(w, v)$.
    Therefore $\tp(u / Mv)$ is not finitely satisfiable in $M$ and $u \nfind {\str M} v$.

    For the other direction assume that $u, v$ are in different connected components of $\gaif(\str N) \setminus \gaif(\str M)$ and take any formula $\phi(u, v) \in \tp(u / Mv)$ with parameters from $M$.
    Denote the quantifier rank of $\phi$ by $q$.

    Pick finite sets $S_u, S_v \subset M$ of vertices reachable from $u$ and $v$ respectively by paths of length at most $7^q$ with all internal vertices outside of $M$.
    Let $\str N'$ be the structure obtained from $\str N$ by isolating all the vertices in $S_u \cup S_v$ and adding for every relation $R$ in the language of $\str N$ all the relations of arity smaller than $R$ that are defined by atomic formulas that use $R$ and constants for $u$ and $v$.
    Informally speaking, these are relations that are obtained by taking $R$ and placing $u$ and $v$ on some positions.
    Clearly, $M$ induces an elementary substructure of $\str N'$, which we will call $\str M'$.
    We can write a formula $\psi(x,y)$ with the parameters from $M$ occurring in $\phi$ and additionally $S_u \cup S_v$ and of the same quantifier rank $q$ as $\phi$, such that
    \begin{align}\label{eqeq}
        \str N \models \phi(x, y) \iff \str N' \models \psi(x, y).
    \end{align}
        
    In $\gaif(\str N')$ the distance between $u$ and $v$ is greater than $7^q$ and the same is true for each of $u$, $v$ and any element of $M$. By \Cref{cor:fo-locality} there is a formula $\alpha(x)$ (without parameters)
    such that for all $w$ whose distance to all parameters in $\psi$ and to $v$ is greater than $7^q$, we have that 
    \begin{align}\label{eq:some-equivalence}
    \str N'\models \alpha(w)\iff \str N'\models \psi(x,v).    
    \end{align}
    
    In particular, $\str N'\models \alpha(u)$.
    Since $\tp_{\str N'}(u/M)$ is finitely satisfiable in $\str M'$ there is an element $w \in M$ such that the distance between $w$ and any parameter in $\psi$ is greater than $7^q$, and both $\str N'\models \alpha(w)$.
    Since $w \in M$ the distance between $w$ and $v$ in $\str N'$ is also greater than $7^q$, by  \eqref{eq:some-equivalence} we get that $\str N' \models \psi(w, v)$, 
    which implies $\str N\models \phi(w,v)$ by \eqref{eqeq}.
    Therefore, $\tp(u/Mv)$ is finitely satisfiable in $\str M$ and $u \find {M} v$.
\end{proof}

\section{Omitted proofs from \Cref{sec:flip-fork}}
\label{app:flip-fork}
In this appendix, we prove \Cref{lem:separation-higher-arity}. We start with the following simple lemma.
\flipsubcommute*
\begin{proof}
    Take a formula $\phi(\tup x)$ in the language of $\str N'$ without parameters.
    Since every relation in $\str N'$ is definable in $\str N$ with parameters from $S$, replace every relation that appears in $\phi(\tup x)$ with its definition in $\str N$, thus obtaining a formula $\psi(\tup x)$ in the language of $\str N$.
    From the construction, we can easily see
    \[
        \str N' \models \phi(\tup n) \iff \str N \models \psi(\tup n) \qquad \text{ for every $\tup n \in N^{\tup x}$}
    \]
    but since $S \subseteq M$, then we also have
    \[
        \str N'[M] \models \phi(\tup m) \iff \str M \models \psi(\tup m) \qquad \text{ for every $\tup m \in M^{\tup x}$}.
    \]
    Moreover, as $\str M \preceq \str N$ then
    \[
        \str M \models \psi(\tup m) \iff \str N \models \psi(\tup m) \qquad \text{ for every $\tup m \in M^{\tup x}$}.
    \]
    As a combination of these implications we get
    \[
        \str N'[M] \models \phi(\tup m) \iff \str N' \models \phi(\tup m) \qquad \text{ for every $\tup m \in M^{\tup x}$},
    \]
    so the statement follows.
\end{proof}

The rest of this section is devoted to a proof of
 \Cref{lem:separation-higher-arity}.
The proof is split into two parts.
First, we show that every ball in the Gaifman graph that has an empty intersection with a model, has bounded number of atomic types over that model.
Then, we show that we can separate elements of every set with a finite number of atomic types over a model by a flip which is definable with parameters from the model.

\subsection{Finite number of types of a ball}

We start with the following lemma which describes a simple obstruction for a structure to be monadically dependent.
\begin{lemma}
    \label{lem:monadic-stability-obstruction}
    Let $\str N$ be a structure and let $\phi(\tup x, y)$ be a formula in a language $\cal L^+$ expanding the language of $\str N$ by additional unary predicates.
    If for every integer $k$ there is a sequence of tuples $\tup m_1, \ldots, \tup m_k \in N^{\tup x}$, a set of elements $A_k = \set{\tup a_{ij} \in N: 1 \le i < j \le k}$, and a monadic lift $\str N^+$ of $\str N$ in the language $\cal L^+$ such that $\str N^{+} \models \phi(\tup m_k, a_{ij}) \iff k \in \set{i, j}$ then $\str N$ is not monadically dependent.
    Moreover, if $\phi$ is an existential formula then $\str N$ is not existentially monadically dependent. 
\end{lemma}
\begin{proof}
    Fix an integer $k$ and take a unary lift $\str N^+$ of $\str N$ that interprets all the unary predicates in the language of $\phi$ such that $\str N^{+} \models \phi(\tup m_k, a_{ij}) \iff k \in \set{i, j}$. Observe that by adding one more unary predicate $U$ to $\str N^+$ that selects a subset $B_k \subseteq A_k$ (thus obtaining a structure $\str N^{++}$) we can interpret arbitrary graph on $k$ vertices with the vertex set given by tuples $\tup m_1, \ldots, \tup m_k$.
    Indeed, for the formula \[\psi(\tup x, \tup y) \equiv \exists z.  U(z) \land \phi(\tup x, z) \land \phi(\tup y, z)\] we have $\str N^{++} \models \psi(\tup m_i, \tup m_j) \iff a_{ij} \in B_k$ (assuming that $i < j$).
    Since $k$ can be arbitrarily large and $\psi$ does not depend on $k$, it follows that $\str N$ is not (existentially) monadically dependent.
\end{proof}

\begin{restatable}{lemma}{fintypes}
    \label{lem:finite-numer-of-types}
    Let $\str M \preceq \str N$ be two existentially monadically dependent structures in a relational language and let $\alpha(\tup w)$ be an atomic formula
    which is stable in $\str N$.
    Fix $r,k\in\N$,
    and let $a \in N \setminus M$ be such that $B_{\str N}^r(a) \cap M = \emptyset$.
    Assume that for every partition $\tup w=x\cup \tup y\cup \tup z$ with $|x|=1$ and $|\tup y| < k$ we have
    \[
        \str N \models \neg \alpha(v, \tup s; \tup s') \; \text{ for all } \; v \in B_{\str N}^r(a), \tup s \in (N \setminus M)^{\tup y}, \tup s' \in M^{\tup z}.
    \]
    Then, for every partition $\tup w=x\cup \tup y\cup \tup z$ with $|x|=1$ and $|\tup y| = k$,  
    \begin{align}\label{eq:finite-types}
        |\Types{\alpha}((B^r(a)^x \times (N \setminus M)^{\tup y})/M^{\tup z})|\ <\ \infty.    
    \end{align}
\end{restatable}

\begin{proof}
    Towards a contradiction, assume that \eqref{eq:finite-types} fails 
     for some partition 
    $\tup w=x\cup \tup y\cup \tup z$ with $|x|=1$ and $|\tup y| = k$.    
    By \Cref{thm:folklore} we can find 
    tuples 
    $\tup v_1, \tup v_2, \ldots \in B^r_{\str N}(a)^x \times (N \setminus M)^{\tup y}$
    and $\tup m_1,\tup m_2,\ldots \in M^{\tup z}$ 
    such that 
    \begin{align}
        \str N\models\alpha(\tup v_i;\tup m_j)\quad\iff\quad i \sim j,\quad\text{for all }i,j\in\N,\label{eq:matching}    
    \end{align}
    where ${\sim}\in\set{=,\neq,\le,\ge}$.
    The cases ${\sim}\in\set{\le,\ge}$ are impossible as $\alpha$ is stable, and the cases ${\sim}\in\set{=,\neq}$ are similar, so we assume $\sim$ is $=$.
    For $i\in\N$, 
      let $\tup p_i$ be the tuple of elements on a shortest path with endpoints $a$ and  
      $\tup v_i[x]\in B^r_{\str N}(a)$
      in the Gaifman graph of $\str N$.
    If the path is of length less that $r$ we extend it with multiple copies of $a$, so we can assume each tuple $\tup p_i$ has length exactly $r$.
    
    Fix a large integer $K$ and denote by $\tup n$ the tuple obtained by concatenating all the tuples $\tup v_i$ and $\tup p_i$, for all $i \in [K]$.
    Denote by $A$ the set of all elements in the tuples $\tup m_1, \ldots, \tup m_K$.
    Apply \Cref{lem:morley-sequence}  to $\tup n$ and $A \subset M$, yielding tuples $\tup n_0, \tup n_1, \tup n_2, \ldots$
    of elements of $M$, of the same length, and atomic type over $A$ as $\tup n$.
    For $j\in\N$ and $i\in [K]$, 
    let $\tup v_{(j, i)}\subseteq \tup n_j$ denote the sub-tuple corresponding 
    to the sub-tuple 
    $\tup v_i\subseteq \tup n$,
     let $\tup p_{{(j,i)}}\subseteq \tup n_j$ 
    denote the sub-tuple corresponding to the sub-tuple $\tup p_i\subseteq \tup n$,
    and let $a_j$ denote the element of $\tup n_j$ that corresponds to the element $a$ of $\tup n$.

    \begin{claim}
    \label{cl:mixed-tuple}
        Let $\tup u\in N^{x\cup \tup y}$ be such that
        \begin{itemize}
            \item for every $z \in x\cup \tup y$ there is some 
            some $i_z\in [K]$ and $j_z\in\N$ with $\tup u[z]=\tup v_{(j_z, i_z)}[z]$; and
            \item $\tup u \not \subseteq \tup n_j$ for every $j\in\N$.
        \end{itemize}
        Then, for every $\tup m_q$ with $q \in [K]$ we have $\str N \not \models \alpha(\tup u, \tup m_q)$.
    \end{claim}
    \begin{claimproof}
        By the second item of \Cref{lem:morley-sequence}, the concatenation of the tuples $\tup n_{j_z}$, for $z\in x\cup\tup y$, has the same atomic type over $A$, regardless of the order of the concatenation.
        We can assume that $j_x\ge j_{z}$ for $z\in \tup y$, and denote $\ell:=j_x$.
        (Otherwise, we can  
        replace $\tup u$ with a different tuple $\tup u'$ satisfying the assumptions of the claim, and such that $\atp(\tup u'/A)=\atp(\tup u/A)$, and for which this assumption holds. Then $\str N\not\models \alpha(\tup u',\tup m_q)$ 
        implies $\str N\not\models \alpha(\tup u,\tup m_q)$, as $\tup m_q\subseteq A$.)

        Then, by the first item of \Cref{lem:morley-sequence} we know that $\tup n_{\ell}$ and $\tup n$ have the same atomic type over $A \cup n_{0} \cup \ldots \cup \tup n_{\ell - 1}$.
        
        Consider the tuple $\tup u'$ which is obtained from $\tup u$ by replacing all the components of $\tup u$ that come from $\tup n_\ell$ with the respective components that come from $\tup n$ instead.
        Observe that $\tup u'[x] \in B^r_{\str N}(a)$ and, because $\tup u\not\subseteq \tup n_\ell$, less than $k=|\tup y|$ other components of $\tup u'$ are in $N \setminus M$.
        Therefore, by the assumption in the statement of the lemma, $\str N \not \models \alpha(\tup u', \tup m_q)$ for every $q \in [K]$.
        As $\tp^{\alpha}(\tup u / A) = \tp^{\alpha}(\tup u' / A)$ we get that $\str N \not \models \alpha(\tup u, \tup m_q)$ for every $q \in [K]$.
    \end{claimproof}

    Denote $L={K\choose 2}$,
    and fix a bijection $f\from[L]\to {[K]\choose 2}$.

    \begin{claim}
        There is a monadic lift $\str N^+$ of $\str N$ and an existential formula $\phi(u, \tup z)$ (independent of $K$),
        such that for every 
        $j\in [L]$ and 
        $i\in [K]$
        we have 
        $$\str N \models \phi(a_j, \tup m_i)\quad\iff\quad i\in f(j).$$
    \end{claim}
    \begin{claimproof}
        The monadic lift $\str N^+$ expands $\str N$ by the following unary predicates.
    For $d\in [r]$,
    the predicates $P_d, P_d'$ are defined by: 
    \[
        P_d \coloneqq \setof{\tup p_{(j,i)}[d]}{j\in [L], i, i'\in f(j), i < i'},
    \]\[
        P_d' \coloneqq \setof{\tup p_{(j,i)}[d]}{j\in [L], i, i'\in f(j), i > i'}.
    \]
    For each $z\in x\cup\tup y$, the predicates
    $S_z, S_z'$ are defined as follows:
    \[
        S_z \coloneqq \setof{\tup v_{(j,i)}[z]}{j\in[L],i, i'\in f(j), i < i'},
    \]\[
        S_z' \coloneqq \setof{\tup v_{(j,i)}[z]}{j\in[L],i, i'\in f(j), i > i'}.
    \]
    We now define the formula $\phi(u,\tup y)$.

    Let $\beta(u,x)$ be the existential formula expressing that there exist $r$ vertices marked with consecutive predicates $P_1, \ldots, P_r$ which form a path in the Gaifman graph joining $u$ and $x$ (maybe with some copies of $u$ at the beginning), and $x$ satisfies the predicate $S_x$.
    Observe that  for any $j\in[L]$ and $v\in N$ we have that $\str N^+\models\beta(a_j,v)$ if and only if $v= \tup v_{(j,i)}[x]$, for some $i, i'\in f(j), i < i'$.
    This is because no internal node of the path $\tup p_{(j, i)}$ is adjacent in the Gaifman graph of $\str N$ to any other $\tup n_h$ for $h \neq j$.
    Indeed, for $h < j$ we know that the $\atp(\tup n_j / A \cup \tup n_h) = \atp(\tup n / A \cup \tup n_h)$.
    Since the internal nodes of $\tup p_i$ are at distance at least $2$ from $M$, then there are no edges in the Gaifman graph between $\tup p_i$ and $\tup n_h$.
    Therefore, the same is true for $\tup p_{(j, i)}$.
    The case $h > j$ follows from the fact that $\atp(\tup n_j\tup n_h) = \atp(\tup n_h\tup n_j)$.

 Next, 
 define 
 $$\gamma(x,\tup z):=\exists \tup y.\alpha(x,\tup y;\tup z)\land S_x(x)\land \bigwedge_{y\in \tup y}S_y(y).$$
    We argue that for $j\in [L]$ and $i\in [K]$
    we have 
    $$\str N^+\models \gamma(\tup v_{(j,i)}[x],\tup m_i)\quad\iff\quad i, i'\in f(j), i < i'.$$
    The right-to-left implication holds, as we 
    can pick the tuple $\tup v_{(j,i)}[\tup y]$ to witnesses the existential quantifiers in $\gamma$.
    Indeed, first, 
    $\str N\models\alpha(\tup v_{i};\tup m_i)$ by \eqref{eq:matching} (recall that $\sim$ is $=$). Next,
    $\str N\models\alpha(\tup v_{(j,i)};\tup m_i)$,
    since $\atp(\tup v_{(j,i)}/A)=
    \atp(\tup v_{i}/A)$ and $\tup m_i\subseteq A$.

    For the left-to-right implication, 
    let $\tup u\in N^{x\cup \tup y}$ be a 
    tuple with $\tup u[x]=\tup v_{(j,i)}[x]$,
    satisfying the formula
    $$\alpha(x,\tup y;\tup m_i) \land S_x(x)\land\bigwedge_{y\in\tup y}S_y(y).$$
    By \Cref{cl:mixed-tuple}, $\tup u\subseteq\tup n_{j'}$ holds for some $j'\in [L]$.
    For $j > h$ we have that $\atp(\tup n_j / \tup n_h) = \atp(\tup n / \tup n_h)$.
    Since $\tup n \subseteq N \setminus M$ and $\tup n_h \subseteq M$ we get that $\tup n_j$ and $\tup n_h$ have disjoint vertex sets.
    As the tuples $(\tup n_{j'}:j'\in [L])$ have pairwise disjoint vertex sets
and $\tup u[x]=\tup v_{(j,i)}[x]$ is an element of $\tup n_j$, it follows that $\tup u\subseteq \tup n_j$. As $\tup u$ 
satisfies the formula
$S_x(x) \land \bigwedge_{y\in\tup y}S_y(y)$ and for each $\tup n_h$ vertices marked with predicates $S_x, S_y$ belong to one tuple $\tup v_{(h, \ell)}$, we get that $\tup u = \tup v_{j, \ell}$ for some $\ell \in [K]$. (A priori it might be $\ell \neq i$ but $\tup v_{(j, i)}[x] = \tup v_{(j, \ell)}[x]$).
However, we know that $\str N \models \alpha(\tup v_{j, \ell}, \tup m_i) \iff \ell = i$, so it must be the case that $\tup u=\tup v_{(j,i)}$.
Finally, if vertices of the tuple $\tup v_{(j, i)}$ are marked with corresponding predicates $S_x, S_y$ it means that $i, i' \in f(j), i < i'$.

    Define 
    $$\psi(u,\tup z):=\exists x.\beta(u,x)\land \gamma(x,\tup z).$$ 

    By the above properties of $\beta$ and $\gamma$ it follows, that $$\str N \models \psi(a_j, \tup m_i)\quad\iff\quad i, i'\in f(j), i < i'.$$

    Define $\psi'(u, \tup z)$ analogically to $\psi$ with predicates $P_d'$ and $S_x', S_y'$ in place of $P_d$ and $S_x, S_y$.
    In the same way we get $$\str N \models \psi'(a_j, \tup m_i)\quad\iff\quad i, i'\in f(j), i > i'.$$

    Finally, define $$\phi(u, \tup z) \coloneqq \psi(u, \tup z) \lor \psi'(u, \tup z).$$
    The statement in claim follows, by the above properties of $\psi$ and $\psi'$.

\end{claimproof}


    As $K$ can be made arbitrarily large and $\phi$ does not depend on $K$, 
    the assumptions of \Cref{lem:monadic-stability-obstruction} 
    are satisfied. Therefore,  $\str N$ is not existentially monadically dependent, a contradiction.
\end{proof}

\subsection{Separating at distance $1$}
In \Cref{lem:one-separation} below we prove a variant of \Cref{thm:forking-for-structures}, for $r=1$.
The proof will produce a concrete form of flips, which we now define.

Flips in the sense of \Cref{def:flip} are a generalization of definable flips in graphs,
 considered in  \Cref{def:flips-graphs}.
According to \Cref{def:flip}, a flip of a graph does not need to be a graph itself.
The following is an analog of \Cref{def:flips-graphs} in relational structures.

\begin{definition}
    \label{def:syntactic-flip}
    Let $\alpha(\tup x;\tup y)$ be a partitioned atomic formula involving a relation $R$.
    We say that $\str M'$ is an \emph{syntactic $(\alpha,S)$-flip} of $\str M$ for some set $S \subseteq M$ if there exists a finite subset $S' \subseteq S$ and $\str M'$ can be obtained from $\str M$ in the following way:
    \begin{enumerate}
        \item Partition  $M^{\tup x}$  according to $\alpha(\tup x;\tup y)$-types over $S'$,
        and $M^{\tup y}$
        according to $\alpha(\tup y;\tup x)$-types over $S'$.
        We will refer to these parts as $(M^{\tup x}, S')$-classes and $(M^{\tup y}, S')$-classes respectively.
        \item Change the interpretation of $R$ in $\str M'$ such that for every $(M^{\tup x}, S')$-class $A$ and every $(M^{\tup y}, S')$-class $B$ we have either \[\str M \models \alpha(\tup a, \tup b) \iff \str M' \models \alpha(\tup a, \tup b)\] for every $\tup a \in A, \tup b \in B$ or \[\str M \models \alpha(\tup a, \tup b) \iff \str M' \not\models \alpha(\tup a, \tup b)\] for every $\tup a \in A, \tup b \in B$.
        
        In the first case we say that the pair $(A, B)$ is not flipped and in the other we say that this pair is flipped. We say that a class is flipped if it is a member of at least one flipped pair.
        \item For each $(M^{\tup x}, S')$-class $A$ that is flipped add to the language of $\str M'$ a new relation symbol $R_A$, interpreted as~$A\subseteq M^{|\tup x|}$. Do the same for each $(M^{\tup y}, S')$-class that is flipped.
        \item Interpretations of all the other relations remain unchanged.
    \end{enumerate}

    We say that a structure $\str N'$ is a \emph{syntactic $S$-flip} of $\str N$ if it can be obtained from $\str N$ by a sequence of such operations.
\end{definition}



Clearly an syntactic $S$-flip is a special case of an $S$-flip in the sense of \Cref{def:flip}.

\medskip

In the proof below, for a formula $\phi(\tup x; \tup y)$ and tuples $\tup a \in N^{\tup x}, \tup b \in N^{\tup y}$, we say that $\tup a$ and $\tup b$ are $\phi$-connected (or $\tup a$ is $\phi$-connected to $\tup b$) if $\str N \models \phi(\tup a; \tup b)$.
We also say that the $\phi$-type of a tuple $\tup c \in N^{\tup x}$ over a set $B$ is empty if $\tup c$ is not $\phi$-connected to any $\tup y$-tuple of vertices from $B$.

\begin{lemma}
    \label{lem:one-separation}
    Let $\str M \preceq \str N$ be two relational structures and let $\alpha(\tup w)$ be a stable atomic formula in their language.
    Fix a set $U \subseteq N \setminus M$ and a partition of $\tup w=x\cup \tup y\cup \tup z$ with $|x|=1$, such that $\Types{\alpha}(U^x \times (N \setminus M)^{\tup y}/ M^{\tup z})$ is finite.
    Then there is a finite set $S\subseteq M$ and an syntactic  $(\alpha(x\tup y; \tup z),S)$-flip $\str N'$ of $\str N$ such that
    \[
        \str N' \models \neg \alpha(v,\tup s; \tup s') \; \text{ for all } \; v \in U^{x}, \tup s \in (N \setminus M)^{\tup y}, \tup s' \in M^{\tup z}.
    \]

    Moreover, if there is a tuple $\tup s \in N^{x \tup y} \setminus M^{x\tup y}$ 
    such that $\str N\models \neg \alpha(\tup s; \tup u)$ for all $\tup u\in  M^{\tup z}$, then the $(S, N^{x\tup y})$-class that corresponds to the empty $\alpha(x\tup y;\tup z)$-type is not flipped.
    The same is true for the $(N^{\tup z},S)$-class that corresponds to the empty $\alpha(\tup z;x\tup y)$-type.

    Assume further that for every partition  $\tup w=x'\cup\tup y'\cup \tup z'$ with $|x'|=1$ and $|\tup y'| < k$ and for all tuples $\tup t \in (N \setminus M)^{\tup y'}$, $\tup t' \in M^{\tup z'}$ and an element $w \in U^{x'}$ we have $\str N\not\models \alpha(w \tup t; \tup t')$.
    Then, for each partition of $\tup w$ into three sets $x''$, $\tup y''$ and $\tup z''$ with $|x''|=1$ and $|\tup y''| \le k$ if $\str N\not\models \alpha(w \tup t; \tup t')$ for all tuples $\tup t \in (N \setminus M)^{\tup y''}$, $\tup t' \in M^{\tup z''}$ and an element $w \in U^{x''}$ then $\str N'\not\models \alpha(w \tup t; \tup t')$ for all tuples $\tup t \in (N \setminus M)^{\tup y''}$, $\tup t' \in M^{\tup z''}$ and an element $w \in U^{x''}$.
\end{lemma}
\begin{proof}
    For every type in $\Types{\alpha}(U^x \times (N \setminus M)^{\tup y}/ M^{\tup z})$ take one tuple from $U^x \times (N \setminus M)^{\tup y}$ with this type and denote them by $\tup b_1, \ldots, \tup b_K$.
    Initialize $S$ as an empty set and for every pair of tuples $\tup b_i, \tup b_j$ with $i < j$ take a tuple $\tup c_{ij} \in M^{\tup z}$ such that $\str N \models \alpha(\tup b_i, \tup c_{ij})$ and $\str N \not \models \alpha(\tup b_j, \tup c_{ij})$ (or vice versa) and add the elements of $\tup c_{ij}$ into $S$.
    Also, for every tuple $\tup b_j$ with a non-empty $\alpha$-type over $M^{\tup z}$ add to $S$ elements of any tuple $\tup c_j$ such that $\str N \models \alpha(\tup b_j, \tup c_j)$.
    By \Cref{thm:definability-of-types}, for every tuple $\tup b_i$ there is a formula $\psi_i(\tup z)$ which is a positive boolean combination of clauses of the form $R(\tup p, \tup z)$ for some parameters $\tup p$ from $M^{x\tup y}$ such that for any $\tup m \in M^{\tup z}$ we have $\str N \models \alpha(\tup b_i, \tup m)$ if and only if $\str N \models \psi_i(\tup m)$.
    Add to $S$ all the parameters from all the formulas~$\psi_i$.

    Now we define $\str N'$.
    Take any $(S, N^{\tup z})$-class $C$.
    We can easily see that every $\tup b_i$ is $\alpha$-homogenous towards $C \cap M^{\tup z}$, i.e. for every $\tup c, \tup c' \in C \cap M^{\tup z}$ we have $\str N \models \alpha(\tup b_i, \tup c) \iff \alpha(\tup b_i, \tup c')$.
    This is because all the tuples in $C$ have the same atomic type over the set of parameters from the formula $\psi_i$.
    Therefore, we can define $\str N'$ by flipping the $(S, M^{x\tup y})$-class of every $\tup b_i$ with these $(S, M^{\tup z})$-classes $C$, such that $\tup b_i$ is $\alpha$-connected to all the tuples in $C \cap M^{\tup z}$.
    It is clear that $\str N' \not\models \alpha(v\tup s; \tup s')$ for every $v \in U^x$, $\tup s \in (N \setminus M)^{\tup y}$, and $\tup s' \in M^{\tup z}$.

    Assume that there is a tuple $\tup s \in N^{x \tup y} \setminus M^{x\tup y}$ 
    such that $\str N\models \neg \alpha(\tup s; \tup u)$ for all $\tup u\in \str M^{\tup z}$.
    Consider the $(S, N^{x\tup y})$-class $T$ consisting of tuples that do not satisfy $\alpha$ in $\str N$ with any parameters from $S$.
    Then, either $T$ doesn't contain any $b_i$ (and therefore is not flipped) or it contains one $b_i$ which is not $\alpha$-connected to any tuple in $M^{\tup z}$ (and again is not flipped).
    On the other hand, if the $(S, N^{\tup z})$-class $T'$ that has empty $\alpha$-type on $S$ is not empty and is flipped, it means that some $\tup b_i$ is $\alpha$-connected to everything in $T' \cap M^{\tup z}$.
    As the neighborhood of $\tup b_i$ is defined by a positive formula, then $\tup b_i$ is $\alpha$-connected to every tuple in $M^{\tup z}$.
    However, this contradicts  \Cref{lem:complete-anticomplete-tuples}.
    Therefore, no empty class is flipped.

    Now consider a partition of the free variables of $\alpha$ into three sets $x''$, $\tup y''$ and $\tup z''$ such that $x''$ is a singleton, $|\tup y''| \le k$ and for all tuples $\tup t \in (N \setminus M)^{\tup y''}$, $\tup t' \in M^{\tup z''}$ and all elements $w \in U^{x''}$ we have $\str N \not\models \alpha(w \tup t; \tup t')$.
    Assume that the partition $x'', \tup y'', \tup z''$ is different than $x, \tup y, \tup z$.

    Take any $w\tup t\tup t' \in U^{x''} \times (N \setminus M)^{\tup y''} \times M^{\tup z''}$.
    Our goal is to show that $\str N' \not \models \alpha(w\tup t; \tup t')$.
    If $x''\tup y'' \subseteq x\tup y$ then the $x\tup y$ part of $w\tup t\tup t'$ is not $\alpha$-connected to anything in $S^{\tup z}$, so in particular it's part is not flipped and since $\str N \not\models \alpha(w \tup t; \tup t')$ the same is true in $\str N'$.
    On the other hand, if $x''\tup y'' \subseteq \tup z$ then again the $\tup z$ part of $w\tup t\tup t'$ is not $\alpha$-connected to $S^{x\tup y}$, so it is not flipped and we get the conclusion in the same way as in the previous case.
    It remains to deal with the cases when $x'' \tup y''$ is split between $x\tup y$ and $\tup z$.
    If $x'' \in x\tup y$ but $x''\tup y'' \not \subseteq x\tup y$ then the $x\tup y$ part of $w\tup t\tup t'$ has one element from $U$ and at most $k - 1$ additional elements from $N \setminus M$, so by our assumption it it is not $\alpha$-connected to any tuple in $S^{\tup z}$ and is not flipped.
    The case when $x'' \in \tup z$ is analogous.
\end{proof}

\subsection{Final separation statement}
We now prove \Cref{lem:separation-higher-arity}, repeated below.
\separationhigherarity*
\begin{proof}
    Denote language of $\str N$ by $\Ll$.
    By \Cref{lem:complete-anticomplete-tuples}, possibly after doing an $\emptyset$-flip, we can assume that for every relation $R(\tup w) \in \Ll$ and for every partition of $\tup w$ into two sets $\tup x, \tup y$, there is no tuple $\tup s \in N^{\tup x} \setminus M^{\tup x}$ such that for every $\tup m \in M^{\tup y}$ we have $\str N \models R(\tup s; \tup m)$.

    Observe, that this lemma is obvious for $a \in M$.
    Therefore, we will assume that $a \in N \setminus M$.

    The idea of this proof is to proceed by induction on $r$ and apply in turns \Cref{lem:finite-numer-of-types} and \Cref{lem:one-separation} thus creating consecutive flips $\str N'$ of $\str N$ (it will be the case that the language of $\str N'$ is not the same as the language of $\str N$ but will always be finite).

    Assume that in the current flip $\str N'$ the ball of radius $r$ around $a$ is disjoint from $M$.
    We iterate over all relations in the language of $\str N'$ starting with the one with highest arity.
    For each such relation $R$ we iterate from $k = 0$ to $k = \text{ar}(R)$ and consider all the tuples in $R$ which have one element from $B^r_{\str N'}(a)$, $k$ elements from $N \setminus M$ and the rest elements from $M$.
    We want to make a flip in which no such tuple for $k' \le k$ is in $R$.
    For this we use \Cref{lem:finite-numer-of-types} which guarantees that the conditions of \Cref{lem:one-separation} are satisfied and therefore we can find a desired flip.

    With each such flip the structure $\str N''$ that we obtain has a finite number of new relations in its language, as it is stated in the definition of an syntactic flip (\Cref{def:syntactic-flip}).
    However, these relations are stable and of smaller arity then $R$.
    Moreover, in every such $S$-flip, the $S$-class that contain $B^r_{\str N'}(a)$ is not flipped and therefore in the Gaifman graph of $\str N''$ the ball $B^r(a)$ remains unchanged, that is, 
    $B^r_{\str N''}(a)=B^r_{\str N'}(a)$.
    Thus, we may now iterate over all newly added relations.
\end{proof}

\highnormality*
\begin{proof}
    We proceed similarly as in the case of graphs (see \Cref{lem:flip-double-separation}).
    Namely, for the same reason as in the case of graphs, the set $S$ constructed in \Cref{lem:one-separation} is at distance at most $2$ from $U$.
    Therefore, if an element $e$ is at distance $r > 3$ from $U$ before the flip, then it is at distance greater than $r - 2$ from $U$ after the flip.
    Since the total number of flips needed for achieving $a \ind {M}r M$ depends only on $r$ and $\str N$, we can separate $a$ form $bM$ by definable flips at a large distance $r'$ and then separate $b$ from $M$ using definable flips at distance $r$ which decreases the distance between $a$ and $bM$ in a controlled way.
\end{proof}

\subsection{Proof of \Cref{prop:forking-for-structures-converse}}
\label{app:converse}
In this section we prove \Cref{prop:forking-for-structures-converse} which we restate below.
\forkingforstructuresconverse*

Before we start the proof, let us remark that we defined the notion of monadic stability for classes of structures.
This is not the usual definition in model theory, where the definition of a \emph{monadically stable theory} is more common.

Let $\Ll$ be a finite relational language and let $\Ll^+$ be $\Ll$ extended with a countably infinite number of unary predicates not present in $\Ll$.
We say that an $\Ll$ theory $T$ is monadically stable if it is stable in the language $\Ll^+$.
Intuitively, we add to our language new unary predicates not mentioned by the theory $T$ and we expect $T$ to be stable in this new language.
By \cite[Lemma 4.4]{braunfeld2022existential}, we get that a class of structures $\Cc$ is monadically stable if and only if its theory $\Th(\Cc)$ is monadically stable.
To be more precise, in the statement of \cite[Lemma 4.4]{braunfeld2022existential} the authors assume that $\Cc$ is a hereditary class of structures, but this assumption can be dropped, as both 
properties ($\CC$ is monadically stable and $\Th(\CC)$ is monadically stable) are preserved by taking the hereditary closure of $\CC$ (the latter follows from \cite[Proposition 2.5]{braunfeld2022existential}).
We remark that the same holds for monadically dependent classes of structures and monadically dependent theories.

Now we state (a fragment of) \cite[Theorem 1.1]{braunfeld2021characterizations}.
We say that a theory $T$ has the \emph{f.s. dichotomy} if the following holds: for all models $\str M \preceq \str N$ of $T$ and all finite tuples $\tup a, \tup b \in N$ such that $\tp(\tup b/M\tup a)$ is finitely satisfiable in $\str M$, for any singleton $c\in N$ we have that either $\tp(\tup b/M\tup ac)$ or $\tp(\tup bc/M\tup a)$ is finitely satisfiable in $\str M$.
\cite[Theorem 1.1]{braunfeld2021characterizations} states in particular that a theory $T$ is monadically dependent if and only if it has the f.s. dichotomy.

\begin{proof}[Proof of \Cref{prop:forking-for-structures-converse}]
    Let $\str M$ be a stable structure such that for every elementary extension $\str M \preceq \str N$, every elementary substructure $\str S \preceq \str N$, and all $a,b\in V(\str N)$
    $$a\find S b\qquad\iff\qquad a\ind S r b\quad\text{for all $r\in\N$}.$$
    Take any model $\str S$ of $\Th(\str M)$ and an elementary extension $\str S \preceq \str N'$.
    By a simple application of compactness, we can find a model $\str N$ of $\Th(\str M)$ such that both $\str M \preceq \str N$ and $\str N' \preceq \str N$.
    Take any finite tuples $\tup a, \tup b \in N$ such that $\tp(\tup b/S\tup a)$ is finitely satisfiable in $\str S$.
    Since $\str S$ is stable, this is equivalent to $\tup a \find S \tup b$.
    By assumption, we have that $\tup a \ind S r \tup b$ for all $r \in \N$.

    Recall that by \Cref{lem:flip-dependence-transitivity}, for any three elements $u, v, w \in N$, if $u \nind S{r} v$ and $v \nind S{q} w$ then $u \nind S{r + q} w$ (this lemma was stated for graphs but the proof works in the same way for structures).
    
    Take any element $c \in N$.
    Observe, that it can't happen that $c \nind S{r} a$ and $c \nind S{q} b$ for some positive integers $r, q$ and elements $a \in \tup a, b \in \tup b$ as that would give us $a \nind S{r + q} b$ and in particular $\tup a \nfind S \tup b$.
    Therefore, we have that either $c \find S \tup a$ or $c \find S \tup b$.
    In the first case we have that $\tup b c \find S \tup a$ and $\tup a c \find S \tup b$ in the second.
    Since forking independence is equivalent to finite satisfiability for stable structures, we get that $\Th(\str M)$ has the f.s. dichotomy, which in turns implies that $\str M$ is monadically stable.
\end{proof}

\end{document}